\documentclass[11pt,reqno]{amsart}
\usepackage{amssymb}
\usepackage{amsmath, amssymb, amsthm}
\usepackage{mathrsfs}

\newtheorem{theorem}{Theorem}[section]
\newtheorem{definition}{Definition}[section]
\newtheorem{lemma}{Lemma}[section]
\newtheorem{remark}{Remark}[section]
\newtheorem{proposition}{Proposition}[section]

\numberwithin{equation}{section}

\newcommand{\supp}{{\mathrm {supp}}}

\begin{document}
\title[Navier-Stokes-Boltzmann equations]{Formation of singularities in solutions to the compressible radiation hydrodynamics equations with vacuum }

\author{yachun li}
\address[Y. C. Li]{Department of Mathematics and Key Lab of Scientific and Engineering Computing (MOE), Shanghai Jiao Tong University,
Shanghai 200240, P.R.China} \email{\tt ycli@sjtu.edu.cn}

\author{Shengguo Zhu}
\address[S. G. Zhu]{Department of Mathematics, Shanghai Jiao Tong University,
Shanghai 200240, P.R.China;
School of Mathematics, Georgia Institute of Technology, Atlanta 30332, U.S.A.}
\email{\tt zhushengguo@sjtu.edu.cn}


\begin{abstract}We study the Cauchy problem for multi-dimensional compressible  radiation hydrodynamics equations with vacuum. First, we present some sufficient conditions on the blow-up of smooth solutions in multi-dimensional space. Then, we obtain the invariance of the support of density for the smooth solutions with compactly supported initial mass density by the property of the system under the vacuum state. Based on the above-mentioned results, we prove that we cannot get a global classical solution, no matter how small the initial data are, as long as the initial mass density is of compact support. Finally, we will see that some of the results that we obtained are still valid for the isentropic flows with degenerate viscosity coefficients as well as for one-dimensional case.
\end{abstract}

\date{September 10, 2012}
\keywords{Navier-Stokes-Boltzmann equations, classical solutions, vacuum, blow-up.\\
{\bf Acknowledgments. } The authors' research was supported in part
by Chinese National Natural Science Foundation under grant 11231006 and 10971135.
}

\maketitle

\section{Introduction}
This paper is concerned with the formation of singularities of smooth solutions to the Cauchy problem for the  compressible radiation hydrodynamics equations with vacuum.

As we know, the effects of heat radiation increase with the growth of temperature. In particular, for the case of thermal equilibrium, radiation energy density is in proportion to the fourth power of temperature. Radiation sometimes contributes largely to energy density, momentum density and pressure, for instance, in astrophysics and inertial confinement fusion. Radiation transfer is usually the most effective mechanism which affects the energy exchange in fluids, so it is necessary to take effects of the radiation field into consideration in the hydrodynamic equations. The equations of radiation hydrodynamics result from the balances of particles, momentum and energy. We first introduce some basic concepts necessary for describing the radiation field and its interaction with matter. At any time $t$, we need $2d$ variables to specify the state of a photon in phase space, namely, $d$ position variables and $d$ velocity (or momentum) variables. Usually, we can denote by $x$ the $d$ position variables, and replace the $d$ momentum variables equivalently with frequency $v$ and the travel direction $\Omega$ of the photon. Via these variables, we define distribution function as
$$ f\equiv f(t,x,v,\Omega),$$then,
$$ dn=f(t,x,v,\Omega)dxdvd\Omega,$$
where $n$ is the number of photons; $dn$ is the number of photons (at time $t$) at space point $x$ in a  volume element $ \text{d}x $, with local frequency $ v $ in a frequency interval $ \text{d}v $, and traveling in a direction $\Omega$ in the cubic angel element $ \text{d}\Omega$.
In the radiation transport, we usually use the specific radiation intensity $ I=I(t,x,v,\Omega)$ to replace the distribution function $f$. The specific radiation intensity is defined as
$$ I=chvf(x,t,v,\Omega)$$
with the Planck constant $h$ and the light speed $c$. The physical interpretation of $I$ is contained in the relationship
$$
dE_1=I(t,x,v,\Omega)\cos \Theta \text{d}\sigma \text{d}v\text{d}\Omega \text{d}t,
$$
where $dE_1$ is the amount of the radiation energy in $\text{d}v$ centered at $v$, traveling in a direction $\Omega$ confined to a solid angel element $\text{d}\Omega$, which crosses, in a time element $\text{d}t$, an area $d\sigma$ oriented such that $\Theta$ is the angel which the direction $\Omega$ makes with the normal to $d\sigma$ ($\cos \Theta=\Omega \cdot \mathbf{n}$, here $\mathbf{n}$ is the outward unit normal vector of $d\sigma$).

Regarding the three basic interactions between photons and matter, namely, absorption, scattering and emission, we have transport equation in the general form
\begin{equation}
\label{eq:1.1}
\frac{1}{c}\partial_tI+\Omega\cdot\nabla I=A_r,
\end{equation}
where
\begin{equation}\label{right}
A_r=\widetilde{S}-\sigma_aI
+\int_0^\infty \int_{S^{d-1}} \left(\frac{v}{v'}\sigma_s(v' \rightarrow v,\Omega'\cdot\Omega)I'
-\sigma_s(v \rightarrow v',\Omega\cdot\Omega')I\right) \text{d}\Omega' \text{d}v',
\end{equation}
$S^{d-1}$ is the unit sphere in $\mathbb{R}^d$,
$$I=I(t,x,v,\Omega),\quad I'=I(t,x,v',\Omega'),$$
$\widetilde{S}=\widetilde{S}(t,x,v,\rho,\theta)$ is the rate of energy emission due to spontaneous process, and $\sigma_a=\sigma_a(t,x,v,\rho,\theta)$ denotes the absorption coefficient that may also depend on the mass density $\rho$ and the temperature $\theta $ of the matter. The dependence of $\sigma_a$ upon $\rho$ and $\theta$ can have the form of (see\cite{gp})
$$ \sigma_a=O(\rho^{\alpha}\theta^{-\beta}), \quad \alpha,  \beta> 0.$$

Similarly to absorption, a photon can undergo scattering interactions with matter, and the scattering interaction serves to change the photon's characteristics $v'$ and $\Omega'$ to a new set of characteristics $v$ and $\Omega$. To quantitatively describe the scattering event, one requires a probabilistic statement concerning this change, which leads to the definition of the `differential scattering coefficient' $\sigma_s(v' \rightarrow v,\Omega'\cdot\Omega)\equiv \sigma_s(v' \rightarrow v, \Omega'\cdot\Omega,\rho,\theta)$ that may depend on $\rho$ and $\theta$ (in general, $ \sigma_s $ is independent of $\theta$)£¬ such that the probability of a photon being scattered from $v'$ to $v$ contained in $\text{d}v$, from $\Omega'$ to $\Omega$ contained in $\text{d}\Omega$, and traveling a distance $ds$ is given by $\sigma_s(v' \rightarrow v,\Omega'\cdot\Omega)\text{d}v \text{d}\Omega ds$. Therefore, the time rates of outscattering and inscattering within a unit volume element are
$$
\text{outscattering}=\int_0^\infty \int_{S^{d-1}} \sigma_s(v \rightarrow v',\Omega\cdot\Omega')I\text{d}\Omega'\text{d}v',
$$
$$
\text{inscattering}=\int_0^\infty \int_{S^{d-1}} \sigma_s(v' \rightarrow v,\Omega'\cdot\Omega)I\text{d}\Omega'\text{d}v',
$$
where $\sigma_s=O(\rho)$. In this paper, for simplicity, we only consider the case $\sigma_s=0$. The local existence of strong solution with nonnegative mass density to the following Navier-stokes-Boltzmann equations for $\sigma_s=O(\rho)$ can be seen in \cite{zll}.  We also prove that the strong solution obtained in \cite{zll} is a classical one for positive time if the initial data has a better regularity, which will be seen in a  forthcoming paper \cite{sz11}.

In the above, we have assumed that $\widetilde{S}$ and $\sigma_a$ are independent of $\Omega$ and $\sigma_s$ depends only on $\Omega\cdot\Omega'$. This means that there exists no inherent preferred direction in the matter.

The impact of radiation on dynamical properties of the fluid is very significant, therefore, we introduce three physical quantities to describe this effect:
\begin{equation}
\label{eq:1.234}
\begin{cases}
\displaystyle
E_r=\frac{1}{c}\int_0^\infty \int_{S^{d-1}} I(t,x,v,\Omega)\text{d}\Omega \text{d}v ,\\[8pt]
\displaystyle
F_r=\int_0^\infty \int_{S^{d-1}}  I(t,x,v,\Omega)\Omega \text{d}\Omega \text{d}v,\\[8pt]
\displaystyle
P_r=\frac{1}{c}\int_0^\infty \int_{S^{d-1}}  I(t,x,v,\Omega)\Omega\otimes\Omega\text{d}\Omega \text{d}v,
\end{cases}
\end{equation}
which are called the radiation energy density, the radiation flux, and the radiation pressure tensor, respectively.

Now we take radiation effect into consideration for viscous fluids to have the following radiation hydrodynamics equations in the $d$-dimensional $(d\geq 2)$ space:
\begin{equation}
\label{eq:1.2}
\begin{cases}
\displaystyle
\frac{1}{c}\partial_tI+\Omega\cdot\nabla I=A_r,\\[10pt]
\displaystyle
\partial_{t}\rho+\nabla\cdot(\rho u)=0,\\[10pt]
\displaystyle
\partial_{t}(\rho u+\frac{1}{c^{2}}F_r)+\nabla\cdot(\rho u\otimes u+P_r)
  +\nabla p_m =\nabla \cdot \mathbb{T},\\[10pt]
\displaystyle
\partial_t(E_m+E_r)+
\nabla\cdot((E_m+p_m)u+F_r)
=\nabla \cdot (u \mathbb{T})+\nabla \cdot (k(\theta)\nabla \theta),
\end{cases}
\end{equation}
where $x\in \mathbb{R}^d$;
$u=(u_1,u_2,...,u_d)\in \mathbb{R}^d$ is the velocity of fluid; $A_r$ is defined by (\ref{right}); $E_m=\frac{1}{2}\rho|u|^2+\rho e$ is the material energy density of the fluid, and $e$ is the internal energy; $k(\theta)\geq 0$ is the heat conductivity; $p_m$ is the material pressure satisfying the equation of state
\begin{equation}
\label{eq:1.3}
p_m=R\rho \theta =\rho^{\gamma}e^S, \quad \gamma > 1,
\end{equation}
where $R$ is a positive constant, $\gamma$ is the adiabatic index and $S$ is the material entropy;
$\mathbb{T}$ is the stress tensor given by
\begin{equation}
\label{eq:1.4}
\mathbb{T}=\mu(\nabla u+(\nabla u)^\top)+\lambda (\nabla\cdot u)\mathbb{I}_d,
\end{equation}
where $\mathbb{I}_d$ is the $d\times d$ unit matrix, $\mu$ is the shear viscosity coefficient, $\lambda$ is the bulk viscosity coefficient, $\mu$ and $\lambda$ are both real constants satisfying
\begin{equation}
\label{eq:1.5}
  \mu \geq 0, \quad \lambda+\frac{2}{d}\mu \geq 0.\end{equation}
From the definition of $\mathbb{T}$, we can easily get
\begin{equation} \label{eq:1.6}   \frac{1}{3}\sum_{i=1}^{d} \mathbb{T}_{ii}=\left(\lambda+\frac{2}{d}\mu\right)\nabla\cdot u.\end{equation}
On one hand, the left-hand side of (\ref{eq:1.6}) is the average normal friction stress caused by viscosity. On the other hand, $ \nabla\cdot u$ on the right-hand side of (\ref{eq:1.6})  represents the rate of change of volume, and $ \nabla\cdot u=\frac{1}{\tau}\frac{d^l\tau}{\text{d}t}$
according to the continuity equation, where $\tau=\frac{1}{\rho}$ is the specific volume and $\frac{d^l}{\text{d}t}=\partial_t+u\cdot \nabla$ stands for the material derivative operator. Thus $ \mu'=\lambda+\frac{2}{d}\mu $
is the ratio between the average normal friction stress and the rate of change of volume, which can be used to describe the change of the average normal friction stress caused by expansion and contraction. We usually call it the second viscosity coefficient. For the monatomic gas, we can take $ \mu'=0 $ when the pressure is not very high. For the diatomic gas like air, if the temperature is not very high, we can also take $ \mu'=0 $. But for usual case, we must consider the effect of $ \mu'$.  It is natural that $ \mu'\geq0 $ from (\ref{eq:1.6}).

We are interested in the blow-up of smooth solutions with initial mass density containing vacuum state.
Makino, Ukai and Kawshima \cite{tms1} discussed the Cauchy problem for the compressible Euler equations and Euler-Possion equations with compactly supported initial density. They
obtained the finite time blow-up for regular solutions which are $C^1$ smooth and satisfy $\rho^{\frac{\gamma-1}{2}}(t,x)\in C^1([0,T)\times\mathbb{R}^3)$ ($\gamma>1$). 
 Via the introduction of the second momentum functional and the estimation of the finite influence domain, they summarized this kind of problems to the solving of  some ODE inequalities. Liu and  Yang \cite{tpy} applied this method to deal with the similar problem of  Euler equations with damping, they first showed that the regular
solutions  will not be global if the initial density function has compact support. Moreover, this method was also applied to viscous heat-conducting compressible flows,  under the assumption that the material entropy is finite in vacuum domain, Cho and Jin \cite{y1} proved the blow-up of smooth solutions in arbitrary space dimensions with  compactly supported initial density.

In this paper, by assuming that the support of mass density grows sub-linearly with time and the initial specific radiation intensity has some directional condition which will be shown in Theorem \ref{th:2.0000}, we will prove in Section $2.3$ the non-global existence of smooth solutions to the compressible radiation hydrodynamics equations with heat conduction via the same method as in \cite{y1}\cite{tms1}. However, since the radiation field affects the mechanical properties of the fluid significantly, it is difficult to get the estimates of some physical quantities. For example, due to the mutual transformation with the radiation energy, we see that the total material energy is not conserved (Remark \ref{r22}). Meanwhile the material momentum $ \int_{\mathbb{R}^d} \rho u \text{d}x $ is also not conserved because of the impact coming from the radiation flux. In order to overcome these difficulties, we have to make full use of the properties of the transport equation (\ref{eq:1.1}) to deal with the additional momentum source term and energy source term generated by the radiation field.

For compressible Navier-Stokes equations with constant viscosity coefficients,  Xin \cite{zx} presented a sufficient condition on the blow-up of smooth solutions in arbitrary space dimensions. He introduced a special functional which is a linear combination of the radial momentum, the second momentum and the total energy. Via the same method, Yang and Zhu \cite{tyc2} proved the non-global existence of the regular solutions with compactly supported initial density and velocity to the one-dimensional isentropic compressible Navier-Stokes equations with degenerate viscosity coefficient. Xin and Yan \cite{zwy} improved the result obtained in \cite{zx}, they got the same blow-up theorem without assuming that the initial density is compactly supported.

In Section $2.2$ of this paper, we first present a sufficient condition on the blow-up of smooth solutions to the compressible radiation hydrodynamics equations with constant viscosity coefficients in multi-dimensional space. The functionals which appear in  \cite{y1}\cite{tms1}\cite{zx}\cite{tyc2} are not valid for our system due to the effects of radiation, e.g., the special growth property of the material entropy (Remark \ref{r31}). We cannot prove the increasing of the material entropy for all the time $t\geq 0$ as we did to Navier-Stokes equations. Actually, in Section $3$, we can only prove the increasing of the material entropy for $t\geq T_c$ ($T_c>0$ is a positive constant) which is different from the case of Navier-Stokes equations. In order to overcome this difficulty, we introduce a new functional which is a linear combination of some mechanic quantities and some radiation quantities, then based on some estimates for the additional terms related to the radiation effects, we get some ODE inequalities which imply the finite time formation of singularities. However, we cannot extend our result directly to the isentropic flow, because we don't have enough estimates on the velocity $u$ in the vacuum domain (Remark \ref{r44}). For the isentropic compressible Navier-Stokes-Boltzmann equations with degenerate viscosity coefficients, in Section $4$, we also use the functional introduced in Section $2.2$ to prove the finite time blow-up of regular solutions in multi-dimensional space, and we point out that our proof is also valid for the same problem to the compressible Navier-Stokes equations in multi-dimensional space.

In Section $5$, we will see that there are some difficulties for the one-dimensional model \cite{BD} \cite{BS}, which is not obtained directly via letting $d=1$ in system (\ref{eq:1.2}). Actually, the system in one-dimensional space is only a certain symmetrization of $m$-d case, then the range of the travel angel variable changes so that we cannot get a valid estimate for the increasing of the material entropy. In order to avoid this difficulty, we assume that the initial specific radiation intensity $I_0$ has some directional conditions as shown in Theorem \ref{th:2.0000} such that we can prove the desired conclusions via the property of the total material energy obtained in Section $2.2$, that is to say, we can also extend some conclusions for multi-dimensional case to one-dimensional case.

In general, studying the  radiation hydrodynamics equations is challenging because of its complexity and mathematical difficulty. There are fewer results for Navier-Stokes-Boltzmann equations or Euler-Boltzmann equations in radiation hydrodynamics. Recently, Jiang-Zhong \cite{sjx} obtained the local existence of $C^1$ solutions for the Cauchy problems of Euler-Boltzmann equations.  Jiang and  Wang \cite{pjd} showed that some $ C^1$ solutions to the compressible Euler-Boltzmann equations will blow up in  finite time, regardless of the size of the initial disturbance. Chen and Wang \cite{zcy} studied the local well-posedness of the Cauchy problems to the  Navier-Stokes-Boltzmann equations under some assumptions. Ducomet and  Ne$\check{\text{c}}$asov$\acute{\text{a}}$ \cite{BD} \cite{BS}  studied the global weak solutions to the Navier-Stokes-Boltzmann equations and its large time behavior for the one-dimensional case.

The rest of  this paper is organized as follows. In Section $2$, we reformulate the
Cauchy problem for system (\ref{eq:1.2}) into a simpler form, and we prove our three main blow-up results.
In Section $3$, we give some applications of blow-up theorems presented in Section $2$. In Section $4$, we give the corresponding
results for the multi-dimensional isentropic flow with degenerate viscosity coefficients. In Section $5$, we remark that some conclusions for multi-dimensional case still hold for one-dimensional case.

\section{Blow-up Conditions for the case of constant viscosity coefficients}
\subsection{Reformulation of the problem}\ \\
Photons are actually bosons. Due to the so called `induced process' of bosons, emission and scattering processes will be enhanced by the photons  in the final state of the reaction. This enhancement can be quantitatively described as follows. If $Z$ represents a basic probability of photon events (emission or scattering), due to the inductive effect, the actual probability is
$$
Z'=Z(1+n_f),
$$
where $ n_f $ stands for the number of  photons  in the final transition  state with the form (see \cite{pjd} or \cite{gp})
$$
n_f=\frac{c^{2}}{2hv^3}I(t,x,v,\Omega).
$$

First we assume that $\sigma_s=0$. From the `induced process'  and the local thermal equilibrium (LTE, see \cite{gp}) assumption, $\widetilde{S}$ and $\sigma_a$ can be written as
\begin{equation*}
\begin{cases}
\widetilde{S}(t,x,v,\rho,\theta)=K_a \overline{B}(v)\left(1+\frac{c^{2}I}{2hv^3}\right),\\[6pt]
\sigma_a(t,x,v,\rho,\theta)=K_a\cdot\left(1+\frac{c^{2}}{2hv^3}\overline{B}(v)\right),\\[6pt]
\end{cases}
\end{equation*}
where $\overline{B}(v)$ is a function of $v$ and $K_a=K_a(t,x,v,\rho,\theta)\geq 0$ ($K_a(t,x,v,0,0)=0$) is the absorption coefficient. More comments on $\widetilde{S}(t,x,v,\rho,\theta)$ and $\sigma_a(t,x,v,\rho,\theta)$ can be seen in Remarks 2.2 and 3.1 in \cite{sjx} as well as in \cite{gp}.
From the state relation
$$ \rho e=\frac{p_m}{\gamma-1}=\frac{1}{\gamma-1}\rho^{\gamma}e^{S},   $$
the energy equation in system (\ref{eq:1.2}) can be reduced to
\begin{equation*}
\ (S_t+u\cdot \nabla S)p_m=N_r,
\end{equation*}
where
$$
N_r=(\gamma-1)\Big(\int_0^\infty \int_{S^{d-1}} \left(1-\frac{u\cdot \Omega}{c}\right)K_a\cdot(I-\overline{B}(v)) \text{d}\Omega \text{d}v
+\nabla \cdot (u \mathbb{T})-u\cdot (\nabla \cdot \mathbb{T})+\nabla \cdot (k(\theta)\nabla\theta)\Big).
$$
So the  Navier-Stokes-Boltzmann  system (\ref{eq:1.2}) ($d\geq2$) can be rewritten into
\begin{equation}
\begin{cases}
\label{eq:2.1}
\displaystyle
\frac{1}{c}\partial_tI+\Omega\cdot\nabla I=-K_a\cdot(I-\overline{B}(v)),\\[8pt]
\partial_{t}\rho+\nabla\cdot(\rho u)=0,\\[8pt]
\displaystyle
\partial_{t}(\rho u)+\nabla\cdot(\rho u \otimes u)
+\nabla p_m =
\frac{1}{c}\int_0^\infty \int_{S^{d-1}} K_a\cdot(I-\overline{B}(v))\Omega \text{d}\Omega \text{d}v+\nabla \cdot \mathbb{T},\\[8pt]
(\partial_tS+u\cdot \nabla S)p_m=N_r.
 \end{cases}
\end{equation}

We consider the Cauchy problem of  (\ref{eq:2.1})  with the initial data
\begin{equation} \label{eq:2.2}
I|_{t=0}=I_0(x,v,\Omega),\quad (\rho,u,S)|_{t=0}=(\rho_0(x),u_0(x), S_0(x))
\end{equation}
satisfying
\begin{equation} \label{eq:2u}
I_0(x,v,\Omega)-\overline{B}(v)\in L^2(\mathbb{R}^+ \times S^{d-1};H^s(\mathbb{R}^d)),\ (\rho_0,u_0,S_0)(x)\in H^s(\mathbb{R}^d),
\end{equation}
\begin{equation} \label{eq:2.2*}
I_0 \geq \overline{B}(v) \ \text{for} \ (x,v,\Omega)\in \mathbb{R}^{d}\times \mathbb{R}^{+}\times S^{d-1},
\end{equation}
\begin{equation} \label{eq:2.2rr}
 I_0\equiv \overline{B}(v) \ \text{for} \  \ |x|\geq R_0, \   (v,\Omega)\in R^+\times S^{d-1},
\end{equation}
where $s> \frac{d}{2}+2$, $R_0$ is a given positive constant, and
$$
\|I_0(x,v,\Omega)-\overline{B}(v)\|^2_{L^2(\mathbb{R}^+ \times S^{d-1};H^s(\mathbb{R}^d))}=\int_0^\infty \int_{S^{d-1}} \|I_0(\cdot,v,\Omega)-\overline{B}(v)\|^2_{H^s(\mathbb{R}^d)}\text{d}\Omega \text{d}v.
$$
\begin{remark}\label{r11}
In this paper, $\overline{B}(v)$ is actually a simplification of the Plank function which represents the energy density of black-body radiation. So condition
$I_0 \geq \overline{B}(v)$ is nature. In Section $2.2$, we will see that the assumption $ I_0\equiv \overline{B}(v)$ for $|x|\geq R_0$ results in the phenomena that the impact of radiation on dynamical properties of the fluid vanishes in the far field, then the system serves as the Navier-Stokes equations as $|x|\rightarrow +\infty$.
\end{remark}
Now we introduce some notations:
\begin{align*}
&m(t)=\int_{\mathbb{R}^{d}}\rho(t,x)\text{d}x \quad \textrm{(total mass)},\\
&M(t)=\int_{\mathbb{R}^{d}} \rho(t,x)|x|^{2}\text{d}x \quad \textrm{ (second moment)},\\
&F(t)=\int_{\mathbb{R}^{d}} \rho(t,x)u(t,x)\cdot x \text{d}x \quad \textrm{ (radial component of momentum)},\\
&\varepsilon(t)=\int_{\mathbb{R}^{d}} E_m(t,x) \text{d}x \quad  \textrm{(total material energy)},\\
&P_m(t)=\int_{\mathbb{R}^{d}}p_m(t,x)\text{d}x \quad  \textrm{ (total material pressure)},
\end{align*}
and
\begin{equation}\label{eq:radiation}
\begin{cases}
\displaystyle
\widetilde{E}_r=\frac{1}{c}\int_0^\infty \int_{S^{d-1}} (I-\overline{B}(v))\text{d}\Omega \text{d}v ,\\[8pt]
\displaystyle
\widetilde{F}_r=\int_0^\infty \int_{S^{d-1}}  (I-\overline{B}(v))\Omega\text{d}\Omega \text{d}v,\\[8pt]
\displaystyle
\widetilde{P}_r=\frac{1}{c}\int_0^\infty \int_{S^{d-1}} (I-\overline{B}(v))\Omega\otimes\Omega \text{d}\Omega \text{d}v.
\end{cases}
\end{equation}
Compared with (\ref{eq:1.234}), we call $\widetilde{E}_r$, $\widetilde{F}_r$, $\widetilde{P}_r$ the correction radiation energy density,  the correction radiation flux and the correction radiation pressure tensor, respectively, purely for technical reason.

We further define
\begin{equation} \label{cd:3}
\begin{cases}
\displaystyle
Q_r(t)=-\frac{2}{c^2}\int_{\mathbb{R}^d} x \cdot \widetilde{F}_r\text{d}x+2(t+\kappa)\int_{\mathbb{R}^d} \widetilde{E}_r\text{d}x,\\[10pt]
I_r(t)=M(t)-2(t+\kappa)F(t)+2(t+\kappa)^{2}\varepsilon(t)+(t+\kappa)Q_r(t),
\end{cases}
\end{equation}
where $ \kappa=\max \{1,\frac{R_0}{c}\}$. $Q_r(t)$ can be regarded as  the combined effect of $E_r(t)$ and $F_r(t)$. We will see later that $I_r(t)>0$.

We always assume that $ m(0),  M(0),  |F(0)|,  \varepsilon(0)< \infty $, and $m(0)> 0$, $\varepsilon(0)> 0$. That is to say, $(I,\rho,u,S)$ is not the trivial zero solution.

\begin{remark}
In this paper, we study the case that $\overline{B}(v)$ only depends on v. The results still hold when the function $\overline{B}(v)$ also depends on $\Omega$. But we do not allow that $\overline{B}(v)$ depends on $x$ or $t$. However, Our further research for the case that
$$
\overline{B}=2hv^3c^{-2}(e^{\frac{hv}{k\theta}-1})^{-1},
$$
which is just the Plank function is in progress. Here $\overline{B}(v)$ depends implicitly on $x$ and $t$ through the temperature  $\theta(t,x)$.
\end{remark}

\subsection{Blow up condition:  the material entropy has lower bound}\ \\
In this section, our main result  is the  estimate (\ref{eq:2.9}) on the life span of smooth solutions to the Cauchy problem (\ref{eq:2.1})-(\ref{eq:2.2}) satisfying  (\ref{eq:2u})-(\ref{eq:2.2*}), which implies that any smooth solution with the support of the density growing sub-linearly in time and the entropy bounded from below cannot exist for all the time. The solution with compactly supported density is a special case in our Theorem. Let
\begin{equation*}\begin{split}
&B_r=\{x\in \mathbb{R}^d||x|\leq r\},\ D(t)=\text{diam}\left(\sup_x \rho(t,x)\right)=\sup_{x,y\in \sup_x \rho(t,x)}|x-y|.\\
&\|I-\overline{B}(v)\|^2_{L^2(\mathbb{R}^+ \times S^{d-1};C^1([0,T);H^s))}=\int_0^\infty \int_{S^{d-1}} \|I(\cdot,\cdot,v,\Omega)-\overline{B}(v)\|^2_{C^1([0,T);H^s)}\text{d}\Omega \text{d}v.
\end{split}
\end{equation*}
\begin{theorem}\label{th:2.1}
Let $T>0$ and
\begin{equation}
\label{eq:1.5jk}
  \mu \geq 0, \quad \lambda+\frac{2}{d}\mu \geq 0,\ k(\theta)\geq 0.\end{equation}
 Suppose that
 $$I(t,x,v,\Omega)\in L^2(\mathbb{R}^+ \times S^{d-1};C^1([0,T);H^s(\mathbb{R}^d))),\ (\rho,u,S)(t,x)\in C^1([0,T);H^s(\mathbb{R}^d))$$
is a smooth solution to Cauchy problem (\ref{eq:2.1})-(\ref{eq:2.2}) satisfying  (\ref{eq:2u})-(\ref{eq:2.2rr}).  Assume further that there exist constants $\alpha$ $(0 \leq \alpha <1)$, $L(> 0)$ and $\underline{S}$ independent of $T $ such that
\begin{equation}\label{cd:1}
D(t)\leq 2L(t+\kappa)^\alpha, \qquad  \forall t \in [0,T),
\end {equation}
\begin{equation}\label{cd:2}
\quad S(t,x)\geq \underline{S} ,    \quad \forall (t,x) \in  [0,T)\times \mathbb{R}^d.
\end {equation}
When $\gamma >  1+\frac{1}{d}$, we need further that $\alpha< \frac{1}{2}$.
Then
\begin{equation}\label{eq:2.9}
T\leq T(\gamma)< +\infty,
\end{equation}
where
$$
T(\gamma)=\left\{\begin{array}{llll}
\left(\frac{I_r(0)}{L_1}\right)^{\frac{1}{d(\gamma-1)(1-\alpha)}},  & 1 < \gamma \leq 1+\frac{1}{d} , \\[8pt]
\left(\frac{I_r(0)}{L_2}\right)^{\frac{1}{1+d\alpha(1-\gamma)}},  & 1+\frac{1}{d}< \gamma < 1+\frac{1}{d\alpha}
\end{array}\right.
$$
with
$$
L_1=\frac{2\kappa^{2-d(\gamma-1)}}{\gamma-1} e^{\underline{S}} L^{(1-\gamma)d} |B_1|^{1-\gamma}m(0)^{\gamma},\
L_2=\frac{2\kappa}{\gamma-1} e^{\underline{S}} L^{(1-\gamma)d} |B_1|^{1-\gamma}m(0)^{\gamma},
$$
and $|B_1|$ is the volume of the unit ball $B_1$. That is to say, any smooth solutions to the Cauchy problem (\ref{eq:2.1})--(\ref{eq:2.2}) has to blow up in finite time as long as (\ref{cd:1})and (\ref{cd:2}) hold.
\end{theorem}

In order to prove Theorem \ref{th:2.1}, we first give some important lemmas.

\begin{lemma}
\label{lemma:2.1}
Let $T> 0$. Assume that
 $$I(t,x,v,\Omega)\in L^2(\mathbb{R}^+ \times S^{d-1};C^1([0,T);H^s(\mathbb{R}^d))),\ (\rho,u,S)(t,x)\in C^1([0,T);H^s(\mathbb{R}^d))$$
 is a smooth solution to Cauchy problem (\ref{eq:2.1})--(\ref{eq:2.2}) satisfying the conditions in Theorem \ref{th:2.1},
 then we have
\begin{equation*}\begin{split}
 I(t,x,v,\Omega) &\geq \overline{B}(v),\quad \forall(t,x,v,\Omega)\in [0,T) \times \mathbb{R}^{d}\times \mathbb{R}^{+}\times S^{d-1};\\
I(t,x,v,\Omega) &\equiv \overline{B}(v),\quad \forall |x|\geq R_0+ct.
\end{split}
\end{equation*}
Moreover, if $ I_0(x,v,\Omega) \equiv \overline{B}(v)$  for $ x \cdot\Omega\leq 0$, then
$$ I(t,x,v,\Omega) \equiv \overline{B}(v), \quad \forall x\cdot\Omega\leq 0.
$$
\end{lemma}

\begin{proof} Because  $\overline{B}(v)$ is independent of $x$ and $t$, the first equation of system (\ref{eq:2.1}) can be rewritten as
$$ \frac{1}{c} \partial_t(I-\overline{B}(v))+\Omega\cdot\nabla (I-\overline{B}(v))=-K_a\cdot(I-\overline{B}(v))$$
We denote by $ y(t;y_0)$ the photon path starting from $y_0$ when $t=0$, i.e.,
      $$   \frac{d}{\partial{t}}y(t;y_0)=c\Omega ,     \qquad  y(0;y_0)= y_0 .                   $$
Along the photon path, we  obtain
     \begin{equation} \label{eq:**}(I-\overline{B}(v))(t,y(t;y_0 ))=(I_0-\overline{B}(v))(y_0 )\exp\Big(\int_0^t -cK_a(\tau,y(\tau;y_0 ),v,\rho,\theta)\text{d}\tau\Big),   \end{equation}
where $y_0=y-c\Omega t $.

Then from (\ref{eq:2.2*}), it is easy to know that
\begin{equation*}\begin{split}
 I(t,x,v,\Omega) &\geq \overline{B}(v),\quad \forall(t,x,v,\Omega)\in [0,T) \times \mathbb{R}^{d}\times \mathbb{R}^{+}\times S^{d-1}.\\
I(t,x,v,\Omega) &\equiv \overline{B}(v),\quad \text{for}\quad |x|\geq R_0+ct.
\end{split}
\end{equation*}
If $ I_0(x,t,\Omega) \equiv \overline{B}(v)$  for $ x \cdot\Omega\leq 0$, we can choose any point $(t,x)\in \mathbb{R}^+\times\mathbb{R}^d  $ satisfying $ x \cdot \Omega \leq 0$, along the photon path
$$x_0\cdot \Omega=(x-c\Omega t) \cdot  \Omega=x \cdot \Omega-c|\Omega|^{2}t \leq x \cdot \Omega \leq 0.  $$
Due to (\ref{eq:**}), it yields
       $$ I(t,x,v,\Omega) \equiv \overline{B}(v), \ \text{for} \ x \cdot \Omega \leq 0. $$
\end{proof}

\begin{lemma}
\label{lemma:2.2} If
$u(x)\in H^s(\mathbb{R}^d)$, then
$ \nabla \cdot (u \mathbb{T})-u\cdot(\nabla \cdot \mathbb{T})\geq 0$.
\end{lemma}
\begin{proof} Calculating directly from the definition of $\mathbb{T}$, we have
\begin{equation*}
\begin{split}
\nabla \cdot (u \mathbb{T})-u\cdot(\nabla \cdot \mathbb{T})&=2\mu \sum_{i=1}^{d} (\partial_{i}u_{i})^{2}+\lambda(\nabla \cdot u)^2\\
&+\mu \sum_{i\neq j}^{d} (\partial_{i}u_{j})^2+2\mu\sum_{i>j} (\partial_{i}u_{j})(\partial_{j}u_{i}).
\end{split}
\end{equation*}
If $ \lambda\leq 0$, according to the Cauchy's inequality, we have
\begin{equation}\label{eq:2.5}
\begin{split}
\nabla \cdot (u \mathbb{T})-u\cdot(\nabla \cdot \mathbb{T}) &\geq (2\mu+d\lambda)\sum_{i=1}^{d}(\partial_{i}u_{i})^{2}+\mu \sum_{i\neq j} (\partial_{i}u_{j})^2+2\mu\sum_{i>j} (\partial_{i}u_{j})(\partial_{j}u_{i})\\
&=(2\mu+d\lambda)\sum_{i=1}^{d}(\partial_{i}u_{i})^{2}+\mu \sum_{i> j} (\partial_{i}u_{j}+\partial_{j}u_{i})^{2}\geq 0.
\end{split}
\end{equation}
 If $ \lambda \geq 0$, it is clear to see that
\begin{equation}\label{eq:2.6}
\begin{split}
\nabla \cdot (u \mathbb{T})-u\cdot(\nabla \cdot \mathbb{T}) &\geq 2\mu \sum_{i=1}^{d}(\partial_{i}u_{i})^{2}+\mu \sum_{i\neq j} (\partial_{i}u_{j})^2+2\mu\sum_{i>j} (\partial_{i}u_{j})(\partial_{j}u_{i})\\
&=2\mu\sum_{i=1}^{d}(\partial_{i}u_{i})^{2}+\mu \sum_{i>j}^{d} (\partial_{i}u_{j}+\partial_{j}u_{i})^{2}\geq 0.
\end{split}
\end{equation}
According to (\ref{eq:1.5}), we have
 $$ \nabla \cdot (u \mathbb{T})-u\cdot(\nabla \cdot \mathbb{T})\geq 0 . $$
\end{proof}

\begin{lemma}
\label{lemma:2.4}Let $T> 0$.
If
$$I(t,x,v,\Omega)\in L^2(\mathbb{R}^+ \times S^{d-1};C^1([0,T);H^s(\mathbb{R}^d))),\ (\rho,u,S)(t,x)\in C^1([0,T);H^s(\mathbb{R}^d))$$
 is a smooth solution to Cauchy problem (\ref{eq:2.1})--(\ref{eq:2.2}) satisfying the conditions in Theorem \ref{th:2.1}, then $$
\varepsilon(t) \geq \varepsilon(0), \ \forall t\in [0,T).
$$
\end{lemma}
\begin{proof}
From the energy equation in (\ref{eq:1.2}), we obtain
\begin{equation*}
\begin{split}
 \partial_t E_m=&
   -(\partial_tE_r+\nabla \cdot F_r )-\nabla \cdot \left(( E_m+p_m)u\right)\\
&+\nabla \cdot (u T)+\nabla \cdot (k(\theta)\nabla \theta).
\end{split}
\end{equation*}
From the transport equation in (\ref{eq:2.1}), we have
$$\partial_t E_r+\nabla \cdot F_r=\int_0^\infty \int_{S^{d-1}}-K_a\cdot(I-\overline{B}(v))\text{d}\Omega \text{d}v.$$
Then, according to (\ref{eq:2.2*}) and Lemma \ref{lemma:2.1}, we have
\begin{equation*}
\begin{split}
\frac{d}{\text{d}t}\varepsilon(t)=\int_{\mathbb{R}^d}\int_0^\infty \int_{S^{d-1}} K_a\cdot(I-\overline{B}(v)) \text{d}\Omega \text{d}v \text{d}x\geq 0.
\end{split}
\end{equation*}
Then
$$ \varepsilon(t) \geq \varepsilon(0), \ \forall t\in [0,T).$$
\end{proof}

\begin{remark}\label{r22}
According to Lemma \ref{lemma:2.4}, we see that total material energy is not conserved because of the radiation effect. This phenomenon is different from the cases in general Euler equations and Navier-Stokes equations since the radiation energy and material energy will be mutually transformed in the process of flow. That is to say, the radiation effect significantly influences the mechanical properties of the fluid. However, the sum of total material energy and total radiation energy is conserved.
\end{remark}

\begin{lemma}
\label{lemma:2.5}Let $T>0$.
If
$$I(t,x,v,\Omega)\in L^2(\mathbb{R}^+ \times S^{d-1};C^1([0,T);H^s(\mathbb{R}^d))),\ (\rho,u,S)(t,x)\in C^1([0,T);H^s(\mathbb{R}^d))$$
is a smooth solution to Cauchy problem (\ref{eq:2.1})--(\ref{eq:2.2}) satisfying the conditions in Theorem \ref{th:2.1},
then
$$  Q_r(t)\geq 0, \quad \forall\ t\in [0,T),$$
where $Q_r(t)$ is defined by (\ref{cd:3}).

\end{lemma}
\begin{proof}From (\ref{eq:radiation})-(\ref{cd:3}), it is easily to have
\begin{equation*}
\begin{split}
Q_r(t)=&\int_{\mathbb{R}^d}\int_0^\infty  \int_{s^{d-1}} \frac{2c(t+\kappa)-2x\cdot\Omega}{c^2}(I-\overline{B}(v))\text{d}\Omega \text{d}v \text{d}x\\
=&\int_{|x|\geq R_0+ct}\int_0^\infty  \int_{s^{d-1}} \frac{2c(t+\kappa)-2x\cdot\Omega}{c^2}(I-\overline{B}(v))\text{d}\Omega \text{d}v \text{d}x\\
&+\int_{|x|\leq R_0+ct}\int_0^\infty  \int_{s^{d-1}} \frac{2c(t+\kappa)-2x\cdot\Omega}{c^2}(I-\overline{B}(v))\text{d}\Omega \text{d}v \text{d}x.
\end{split}
\end{equation*}
Due to $\kappa=\max \{1,\frac{R_0}{c}\}$, when $|x|\leq R_0+ct$,
$$ 2c(t+\kappa)-2x\cdot\Omega \geq 2ct+2c\kappa-2R_0-2ct \geq 0.$$
From Lemma \ref{lemma:2.1}, we can easily get
\begin{equation*} Q_r(t) \geq 0, \quad \forall\ t\in [0,T).
\end{equation*}
\end{proof}
Now we start the proof of Theorem \ref{th:2.1} from the following key estimate on the total material pressure:
\begin{proposition}[\textbf{The behavior of material pressure}]\label{pro:2.1}\ \\
Let $T>0$ and $p_m(x,t)$ be the material pressure associated with the solution
$$I(t,x,v,\Omega)\in L^2(\mathbb{R}^+ \times S^{d-1};C^1([0,T);H^s(\mathbb{R}^d))),\ (\rho,u,S)(t,x)\in C^1([0,T);H^s(\mathbb{R}^d))$$
to  Cauchy problem (\ref{eq:2.1})--(\ref{eq:2.2}) satisfying the conditions in Theorem \ref{th:2.1},
 then we have
\begin{equation*}
P_m(t)\leq\left\{\begin{array}{llll}
\frac{\gamma-1}{2\kappa^{2-d(\gamma-1})}(t+\kappa)^{-(\gamma-1)d}I_r(0),  &1 < \gamma \leq 1+\frac{1}{d}, \\[10pt]
\frac{\gamma-1}{2}(\kappa(t+\kappa))^{-1}I_r(0),   & 1+\frac{1}{d}< \gamma < +\infty
\end{array}\right.
\end{equation*}
for all $t\in [0,T)$.
\end{proposition}
\begin{proof}
The definition  of $I_r(t)$ in (\ref{cd:3}) and direct calculations lead to
\begin{equation}\label{def}
\begin{split}
I_r(t)
&=\int_{\mathbb{R}^d} |x-(t+\kappa)u|^{2}\rho \text{d}x+\frac{2}{\gamma-1}(t+\kappa)^{2}P_m(t)+(t+\kappa)Q_r(t).
\end{split}
\end{equation}
 Integrating by parts and using system (\ref{eq:2.1}), we get
\begin{equation}\label{eq:3.1}
\begin{split}
\frac{d}{\text{d}t}I_r(t)=&-2(t+\kappa)\int_{\mathbb{R}^d} (\rho |u|^2+dp_m ) \text{d}x
+4(t+\kappa)\int_{\mathbb{R}^d} E_m \text{d}x\\
&+2(t+\kappa)\int_{\mathbb{R}^d}x\cdot (\nabla \cdot \widetilde{P}_r)\text{d}x +4(t+\kappa)\int_{\mathbb{R}^d}\widetilde{E}_r \text{d}x
-\frac{2}{c^2}\int_{\mathbb{R}^d} x\cdot \widetilde{F}_r\text{d}x.
\end{split}
\end{equation}
Noticing that $\widetilde{P}_r=(\widetilde{P}^{ij}_r)_{d\times d}$ is a tensor of  order $d$,
where
$$
\widetilde{P}^{ij}_r=\frac{1}{c}\int_0^{\infty}\int_{S^{d-1}}(I-\overline{B}(v))\Omega_{i}\Omega_{j}\text{d}\Omega \text{d}v,
$$
we have
\begin{equation*}
\begin{split}
\nabla \cdot (x\cdot \widetilde{P}_r )
&=\nabla \cdot \left(\sum_{i=1}^{d}x_{i}\widetilde{P}^{i1}_r,\sum_{i=1}^{d}x_{i}\widetilde{P}^{i2}_r,....,\sum_{i=1}^{d}x_{i}\widetilde{P}^{id}_r\right)^T
=\sum_{i=1}^{d}\sum_{j=1}^{d}\left( x_{i}\frac{\partial{\widetilde{P}^{ij}_r}}{\partial{x_j}}+\delta_{ij}\widetilde{P}^{ij}_r\right)\\
&=x\cdot (\nabla \cdot \widetilde{P}_r) +\sum_{i=1}^{d}\widetilde{P}^{ii}_r
=x\cdot (\nabla \cdot \widetilde{P}_r)+\widetilde{E}_r,
\end{split}
\end{equation*}
where $\delta_{ij}$ is the Konecker symbol satisfying $\delta_{ij}=1$, for $ i=j$; $\delta_{ij}=0$, for  $ i\neq j $.

Noting that
$$
E_m=\frac{1}{2}\rho|u|^2+\rho e,\ \rho e=\frac{p_m}{\gamma-1}, $$
from (\ref{eq:3.1}) and integrating by parts, we have
\begin{equation}\label{cd:4}
\begin{split}
\frac{d}{\text{d}t}I_r(t)=&\frac{2}{\gamma-1}(2-d(\gamma-1))(t+\kappa)P_m(t)\\
&+\int_{\mathbb{R}^d}\int_0^\infty  \int_{s^{d-1}} \frac{2c(t+\kappa)-2x\cdot\Omega}{c^2}(I-\overline{B}(v))\text{d}\Omega \text{d}v \text{d}x\\
=&\frac{2} {\gamma-1} (2-d(\gamma-1))(t+\kappa) P_m(t)+Q_r(t).
\end{split}
\end{equation}
From the definition of $ I_r(t) $, we know that
\begin{equation*}
\begin{split}
\frac{2-d(\gamma-1)}{t+\kappa}I_r(t)=& \frac{2}{\gamma-1}(2-d(\gamma-1))\int_{\mathbb{R}^d} |x-u(t+\kappa)|^{2}\rho \text{d}x\\
&+\frac{2} {\gamma-1} (2-d(\gamma-1))(t+\kappa)P_m(t)+(2-d(r-1))Q_r(t).
\end{split}
\end{equation*}
According to  Lemmas \ref{lemma:2.4}-\ref{lemma:2.5},  when $1< \gamma \leq 1+\frac{2}{d}$, from (\ref{cd:4}) we have
\begin{equation}\label{eq:3.3f}
\frac{d}{\text{d}t}I_r(t) \leq \frac{2-d(\gamma-1)}{t+\kappa}\eta I_r(t).
\end{equation}
where $\eta=\max\big\{1,\frac{1}{2-d(\gamma-1)}\big\}$.
Integrating (\ref{eq:3.3f}) yields
\begin{equation}\label{eq:3.3ff}
I_r(t) \leq I_r(0) \Big(\frac{t+\kappa}{\kappa}\Big)^{\eta(2-d(\gamma-1))}.
\end{equation}
If $1<\gamma \leq1+\frac{1}{d}$, then $\eta=1$, from (\ref{eq:3.3f}) we get the first estimate
\begin{equation}\label{eq:3.4q}
P_m(t)\leq \frac{\gamma-1}{2\kappa^{2-d(\gamma-1)}}(t+\kappa)^{-(\gamma-1)d}I_r(0).
\end{equation}
If $1+\frac{1}{d}< \gamma \leq 1+\frac{2}{d}$, then $\eta=\frac{1}{2-d(\gamma-1)}$, from (\ref{eq:3.3f}) we get the second estimate
\begin{equation}\label{eq:3.4qq}
P_m(t)\leq \frac{\gamma-1}{2}(\kappa(t+\kappa))^{-1}I_r(0).
\end{equation}
If  $\gamma> 1+\frac{2}{d}$, due to $2-d(\gamma-1)< 0 $, from (\ref{cd:4}) we have
\begin{equation*}
\begin{split}
\frac{d}{\text{d}t}I_r(t) \leq Q_r(t)
\leq \frac{1}{t+\kappa}I_r.
\end{split}
\end{equation*}
Solving this inequality, we have
\begin{equation*}
I_r (t)\leq \frac{I_r(0)}{\kappa}(t+\kappa),
\end{equation*}
and thus
\begin{equation}\label{eq:3.4qqq}
P_m(t)\leq \frac{\gamma-1}{2}(\kappa(t+\kappa))^{-1}I_r(0).
\end{equation}
Then from (\ref{eq:3.4q})-(\ref{eq:3.4qqq}), we get the desired estimates.\\

\end{proof}
\begin{remark}We emphasize that
proposition \ref{pro:2.1} holds without the additional conditions (\ref{cd:1}) and (\ref{cd:2}) as long as $M(t)$ is well-defined.
\end{remark}
Now we prove  Theorem \ref{th:2.1}.
\begin{proof} \textbf{Case 1}: $1 < \gamma < 1+\frac{2}{d}$.\\
From the proof of Proposition \ref{pro:2.1},
we have
\begin{equation}\label{eq:3.5}
I_r(0) \geq \frac{2\kappa^{\eta(2-d(\gamma-1))}}{\gamma-1}(t+\kappa)^{2-\eta(2-d(\gamma-1))}\int_{\mathbb{R}^d} p_m \text{d}x.
\end{equation}

From Jensen's inequality, we have
\begin{equation*}
\begin{split}
I_r(0) &\geq  \frac{2\kappa^{\eta(2-d(\gamma-1))}}{\gamma-1}(t+\kappa)^{(2-\eta(2-d(\gamma-1)))} e^{\underline{S}} |\supp_x \rho(t,x)|  \int_{\supp_x \rho(t,x)} (\rho(t,x))^\gamma \frac{\text{d}x}{|\supp_x \rho(t,x)|}\\
&\geq \frac{2\kappa^{\eta(2-d(\gamma-1))}}{\gamma-1}(t+\kappa)^{(2-\eta(2-d(\gamma-1)))} e^{\underline{S}} |\supp_x \rho(t,x)| ^{1-\gamma}m(0)^{\gamma}
\\
&\geq \frac{2\kappa^{\eta(2-d(\gamma-1))}}{\gamma-1}e^{\underline{S}}L^{(1-\gamma)d}|B_1|^{1-\gamma}(t+\kappa)^{(2-\eta(2-d(\gamma-1)))+d\alpha(1-\gamma)}m(0)^{\gamma}
\\
&\equiv L_\gamma(t+\kappa)^{(2-\eta(2-d(\gamma-1))+d\alpha(1-\gamma))},
\end{split}
\end{equation*}
where $|\supp_x \rho(t,x)|$ is the volume of $\supp_x \rho(t,x)$, $L_\gamma=\frac{2\kappa^{\eta(2-d(\gamma-1))}}{\gamma-1}e^{\underline{S}}L^{(1-\gamma)d}|B_1|^{1-\gamma}m(0)^{\gamma}$, and we used the fact that

\begin{equation}\label{eq:3.6}
\int_{\supp_x \rho(t,x)} \rho(t,x) \text{d}x=\int_{\mathbb{R}^d} \rho(t,x) \text{d}x=\int_{\mathbb{R}^d} \rho_0(x) \text{d}x= m(0),
\end{equation}
which can be obtained easily from
\begin{equation}\label{eq:3.7}
\partial_t \rho(t,x)=\int_{\mathbb{R}^d} -\nabla\cdot(\rho u)\text{d}x=0.
\end{equation}
If $1<\gamma \leq1+\frac{1}{d}$, then $\eta=1$, we have
\begin{equation}\label{ddd1}
I_r(0) \geq L_\gamma(t+\kappa)^{d(1-\alpha)(\gamma-1)}, \ L_\gamma=L_1=\frac{2\kappa^{2-d(\gamma-1)}}{\gamma-1}e^{\underline{S}}L^{(1-\gamma)d}|B_1|^{1-\gamma}m(0)^{\gamma}.
\end{equation}
If $1+\frac{1}{d}< \gamma \leq 1+\frac{2}{d}$, then $\eta=\frac{1}{2-d(\gamma-1)}$, we have
\begin{equation}\label{ddd2}
I_r(0) \geq L_\gamma(t+\kappa)^{(1+d\alpha(1-\gamma))}, \  L_\gamma=L_2=\frac{2\kappa}{\gamma-1}e^{\underline{S}}L^{(1-\gamma)d}|B_1|^{1-\gamma}m(0)^{\gamma}.
\end{equation}
\textbf{Case 2}: $ \gamma \geq 1+\frac{2}{d}$. We have
\begin{equation}\label{ddd3}
\begin{split}
I_r(0) \geq & \frac{2\kappa}{\gamma-1}(t+\kappa)\int_{\supp_x \rho(t,x)} (\rho(t,x))^\gamma \text{d}x
\geq  L_2(t+\kappa)^{(1+d\alpha(1-\gamma))}.
\end{split}
\end{equation}
(\ref{eq:2.9}) follows immediately from (\ref{ddd1})-(\ref{ddd3}).\\


\end{proof}

\subsection{Blow up condition:  the material entropy does not have lower bound}\ \\
In many cases, if we consider the heat conduction, then the lower bound of entropy is not easy to get. Next, we will give a blow-up result without the condition that the entropy is bounded from below. Let
\begin{equation*}
R(t)=\inf\{r|\supp_x \rho(t,x) \subseteq B_r\}.
\end{equation*}
\begin{theorem}\label{th:2.0000}
Let $T>0$ and
\begin{equation}
\label{eq:1.5jk}
  \mu \geq 0, \quad \lambda+\frac{2}{d}\mu \geq 0,\quad  k(\theta)\geq 0.\end{equation}
Suppose that
\begin{equation*}I(t,x,v,\Omega)\in L^2(\mathbb{R}^+ \times S^{d-1};C^1([0,T);H^s(\mathbb{R}^d))),\ (\rho,u,S)(t,x)\in C^1([0,T);H^s(\mathbb{R}^d))\end{equation*}
is a smooth solution to the Cauchy problem (\ref{eq:2.1})-(\ref{eq:2.2}) satisfying  (\ref{eq:2u})-(\ref{eq:2.2*}) and
\begin{equation}\label{bbbb}
I_0\equiv \overline{B}(v), \  \text{for} \quad  x\cdot \Omega \leq 0.
\end{equation}
Assume further that there exist constants $\alpha$ $(0 \leq \alpha <1)$ and  $\widetilde{L}(> 0)$ such that
\begin{equation}\label{cd:11}
R(t)\leq \widetilde{L}(t+1)^\alpha, \  \forall \  t \in [0,T).
\end {equation}
Then
\begin{equation}\label{eq:2.91}
T< +\infty.
\end{equation}
That is to say, any smooth solution to the Cauchy problem (\ref{eq:2.1})--(\ref{eq:2.2}) has to blow up in finite time as long as  (\ref{bbbb}) and (\ref{cd:11}) hold.
\end{theorem}

\begin{proof}
We will use those physical quantities defined in Section 2.1. From the continuity equation and integrating by parts, we get
\begin{equation}\label{eq:3.8}
\frac{d}{\text{d}t}M(t)=2F(t).
\end{equation}

From the momentum equations and integrating by parts, we also get
\begin{equation}\label{eq:3.9}
\frac{d}{\text{d}t}F(t)=\int_{\mathbb{R}^d} (\rho|u|^2 +d p_m ) \text{d}x+\frac{1}{c}\int_{\mathbb{R}^d}\int_0^\infty \int_{S^{d-1}} K_a\cdot(I-\overline{B}(v)) x\cdot\Omega  \text{d}\Omega \text{d}v\text{d}x.
\end{equation}

From Lemma \ref{lemma:2.1} and (\ref{bbbb}), we know that
\begin{equation}\label{eq:3.10}
\begin{split}
&\frac{1}{c}\int_{\mathbb{R}^d}\int_0^\infty \int_{S^{d-1}} K_a\cdot(I-\overline{B}(v))x\cdot\Omega \text{d}\Omega  \text{d}v \text{d}x \\
=&\frac{1}{c}\int_{\mathbb{R}^d}\int_0^\infty \int_{S^{d-1}\bigcap \{x\cdot \Omega \geq 0\}}  K_a\cdot(I-\overline{B}(v))x\cdot\Omega \text{d}\Omega  \text{d}v\text{d}x\\
&+\frac{1}{c}\int_{\mathbb{R}^d}\int_0^\infty \int_{S^{d-1}\bigcap \{x\cdot \Omega < 0\}}  K_a\cdot(I-\overline{B}(v))x\cdot\Omega \text{d}\Omega  \text{d}v\text{d}x
\geq 0.
\end{split}
\end{equation}
Combining (\ref{eq:3.9}) and (\ref{eq:3.10}), we arrive at
\begin{equation} \label{eq:3.11}\frac{d}{\text{d}t}F(t)\geq \int_{\mathbb{R}^d} (\rho|u|^2 +d p_m ) \text{d}x \geq 0. \end{equation}
Integrating (\ref{eq:3.8}) and (\ref{eq:3.11}) over $[0,t]$, respectively, we obtain
\begin{equation}\label{eq:3.12}
M(t)=M(0)+2\int_{0}^{t} F(\tau)\text{d}\tau,
\end{equation}
and
\begin{equation}\label{eq:3.13}
F(t)\geq F(0)+\int_{0}^{t}\int_{\mathbb{R}^d}\rho |u|^2(\tau,x)\text{d}x \text{d}\tau+d\int_{0}^{t}\int_{\mathbb{R}^d} p_m(\tau,x)\text{d}x\text{d}\tau.
\end{equation}

In the case $1< \gamma < 1+\frac{2}{d}$, using the definition of $E_m$, together with Lemma \ref{lemma:2.4} , we  have
\begin{equation}\label{eq:3.1533}
\begin{split}
F(t)&\geq F(0)+2\int_{0}^{t} \varepsilon (\tau)\text{d}\tau+\Big(d-\frac{2}{\gamma-1}\Big)\int_{0}^{t} \int_{\mathbb{R}^d}p_m(\tau,x) \text{d}x \text{d}\tau\\
&= F(0)+2\int_{0}^{t} \varepsilon (\tau)\text{d}\tau+\Big(d-\frac{2}{\gamma-1}\Big)\int_{0}^{t}\int_{\mathbb{R}^d}(\gamma-1)(\rho e)(\tau,x)\text{d}x\text{d}\tau\\
&\geq F(0)+2\int_{0}^{t} \varepsilon (\tau)\text{d}\tau+\Big(d-\frac{2}{\gamma-1}\Big)\int_{0}^{t}(\gamma-1) \varepsilon (\tau)\text{d}\tau\\
&\geq F(0)+d(\gamma-1)\int_{0}^{t} \varepsilon (\tau)\text{d}\tau  \geq F(0)+d(\gamma-1) \varepsilon (0)t.
\end{split}
\end{equation}
Substituting (\ref{eq:3.1533}) into (\ref{eq:3.12}), we have
\begin{equation}\label{eq:3.1733}
M(t)\geq M(0)+2 F(0)t+d(\gamma-1)\varepsilon(0)t^2.
\end{equation}
Using condition (\ref{cd:11}), it yields
\begin{equation}\label{eq:3.1833}
M(t)=\int_{\mathbb{R}^d} \rho |x|^2 \text{d}x \leq  R^{2}(t)m(0) \leq \widetilde{L}^{2}(1+t)^{2\alpha}m(0).
\end{equation}
Combining (\ref{eq:3.1733}) and (\ref{eq:3.1833}), we have
\begin{equation*}
  \widetilde{L}^{2}(t+1)^{2\alpha}\geq M(0)+2 F(0)t+d(\gamma-1)\varepsilon(0)t^2.
\end{equation*}
Because of $0 \leq \alpha < 1 $, we must have
$$ T <+\infty.$$
In the case $1+\frac{2}{d}\leq \gamma < +\infty$, using the definition of $E_m$, together with Lemma \ref{lemma:2.4} , we  have
\begin{equation}\label{eq:3.15}
\begin{split}
F(t)&\geq F(0)+\int_{0}^{t} \int_{\mathbb{R}^d} \rho |u|^2(\tau,x)\text{d}x \text{d}\tau+ d\int_{0}^{t} \int_{\mathbb{R}^d}(\gamma-1)\rho e (\tau,x) \text{d}x\text{d}\tau\\
&=F(0)+2\int_{0}^{t} \varepsilon (\tau)\text{d}\tau+\Big(d-\frac{2}{\gamma-1}\Big)\int_{0}^{t} \int_{\mathbb{R}^d}p_m(\tau,x)\text{d}x\text{d}\tau\\
&\geq F(0)+2\int_{0}^{t} \varepsilon (0)\text{d}\tau+\Big(d-\frac{2}{\gamma-1}\Big)\int_{0}^{t}\int_{\mathbb{R}^d} p_m(\tau,x)\text{d}x\text{d}\tau\\
&\geq F(0)+2\varepsilon (0)t.
\end{split}
\end{equation}
Substituting (\ref{eq:3.15}) into (\ref{eq:3.12}), we have
\begin{equation}\label{eq:3.17}
M(t)\geq M(0)+2 F(0)t+2\varepsilon(0)t^2.
\end{equation}
Combining (\ref{eq:3.1833}) and (\ref{eq:3.17}), we have
\begin{equation}
  \widetilde{L}^{2}(t+\kappa)^{2\alpha}\geq M(0)+2 F(0)t+2\varepsilon(0)t^2.
\end{equation}
Because of $0 \leq \alpha < 1 $,  we must have
$$ T <+\infty. $$
\end{proof}

\subsection{Blow up condition:  the velocity  fastly decays}\ \\
  The next blow-up condition tells us that there does not exist global smooth solution with the velocity $u$ with a little bit fast decay as follows as time goes on:

\begin{theorem}\label{th:2.20}
Let
\begin{equation}
\label{eq:1.5jk}
  \mu \geq 0, \quad \lambda+\frac{2}{d}\mu \geq 0,\ k(\theta)\geq 0.\end{equation}
There is no global smooth solution
$$ I(t,x,v,\Omega)\in L^2(\mathbb{R}^+ \times S^{d-1};C^1([0,T);H^s(\mathbb{R}^d)),\ (\rho,u,S)(t,x)\in C^1([0,+\infty);H^s(\mathbb{R}^d))$$
satisfying
\begin{equation}\label{eq:2.15}
\limsup_{t\rightarrow +\infty} \left|\left| \frac{t}{1+|x|^2} u(t,x)\cdot x \right|\right|_{L^\infty}<1
\end{equation}
to the Cauchy problem (\ref{eq:2.1})--(\ref{eq:2.2}) satisfying (\ref{eq:2u})-(\ref{eq:2.2*}) and (\ref{bbbb}).
\end {theorem}

\begin{proof}
Let
$$I(t,x,v,\Omega)\in L^2(\mathbb{R}^+ \times S^{d-1};C^1([0,+\infty);H^s(\mathbb{R}^d))),\ (\rho,u,S)(t,x)\in C^1([0,+\infty);H^s(\mathbb{R}^d))$$ be a smooth solution to Cauchy problem (\ref{eq:2.1})--(\ref{eq:2.2}) satisfying (\ref{eq:2u})-(\ref{eq:2.2*}) and (\ref{bbbb}). Then there exists constants $t_0 > 0$ and $C_0 < 1$ such that for all $ t\geq t_0$, \\
\begin{equation*}
\left|\left|\frac{u(t,x)\cdot x}{1+|x|^2} \right|\right|_{L^\infty}<\frac{C_0}{t}.
\end{equation*}

Let $\widetilde{M}(t)=\int_{\mathbb{R}^d} \rho(1+|x|^2)\text{d}x $ . Then
\begin{equation*}
\frac{d}{\text{d}t}\widetilde{M}(t)=2\int_{\mathbb{R}^d} \rho x\cdot u \text{d}x \leq 2 \widetilde{M}(t)\left|\left| \frac{u(t,x)\cdot x}{1+|x|^2} \right|\right|_{L^\infty}\leq 2 C_0 \frac{\widetilde{M}(t)}{t}
\end{equation*}
for all $ t\geq t_0$.

Solving this inequality, we have

\begin{equation}\label{eq:4.4}
\widetilde{M}(t) \leq \widetilde{M}(t_0)\exp\Big(2C_0 \ln {\frac{t}{t_0}}\Big)=\frac{\widetilde{M}(t_0)}{t^{2C_0 }_0}t^{2C_0}=\frac{M(t_0)+m(0)}{t^{2C_0 }_0}t^{2C_0}.
\end{equation}

According to the momentum equations and integrating by parts, we have
\begin{equation}\label{eq:4.5}
\frac{d^2}{\text{d}t^2}\widetilde{M}(t)=2\int_{\mathbb{R}^d} (\rho |u|^2 +d p_m) \text{d}x +\frac{1}{c}\int_{\mathbb{R}^d}\int_0^\infty \int_{S^{d-1}} K_a\cdot(I-\overline{B}(v))x\cdot\Omega \text{d}\Omega  \text{d}v\text{d}x,
\end{equation}
where
\begin{equation}\label{eq:4.6}
\frac{1}{c}\int_{\mathbb{R}^d}\int_0^\infty \int_{S^{d-1}} K_a\cdot(I-\overline{B}(v))x\cdot\Omega \text{d}\Omega  \text{d}v\text{d}x\geq 0,
\end{equation}
which can be seen in the proof of Theorem 2.2.

According to $
p_m=(\gamma-1)\rho e
$
and Lemma \ref{lemma:2.4}, we know that
\begin{equation*}
\begin{split}
\frac{d^2}{\text{d}t^2}\widetilde{M}(t)& \geq 2\int_{\mathbb{R}^d} (\rho |u|^2 +dp_m)(\tau,x) \text{d}x \\
=&4\int_{0}^{t} \varepsilon (\tau)\text{d}\tau+2\Big(d-\frac{2}{\gamma-1}\Big)\int_{0}^{t} \int_{\mathbb{R}^d}p_m(\tau,x)\text{d}x\text{d}\tau\\
=&4\int_{0}^{t} \varepsilon (\tau)\text{d}\tau+2\Big(d-\frac{2}{\gamma-1}\Big)\int_{0}^{t}\int_{\mathbb{R}^d}(\gamma-1)(\rho e)(\tau,x)\text{d}x\text{d}\tau\\
\geq& \eta \varepsilon (t) \geq  \eta \varepsilon(0),
\end{split}
\end{equation*}
where $\eta=\min\{4,2d(\gamma-1)\} > 0$. From the Taylor's expansion, we have
\begin{equation}\label{eq:4.7}
\widetilde{M}(t) \geq m(0)+M(0)+F(0)t+\frac{1}{2}\eta \varepsilon(0)t^2.
\end{equation}
Combining (\ref{eq:4.4}) and (\ref{eq:4.7}), we get
\begin{equation}\label{**}
m(0)+M(0)+F(0)t+\frac{1}{2}\eta \varepsilon(0)t^2 \leq \frac{M(t_0)+m(0)}{t^{2C_0 }_0}t^{2C_0}
\end{equation}
for all  $ t\geq t_0$.

On the other hand, since $ 2C_0 < 2 $ and $\eta \varepsilon(0) > 0$, (\ref{**}) fails to hold when $t$ is large enough.
This contradiction implies that such a solution does not exist.
\end{proof}

\section{Some applications  }

In this section, we shall give several applications of the blow-up conditions presented in Section 2. First, we point out that the sub-linear growth condition (\ref{cd:1}) on the support of the density can be verified for most flows. We have the  following invariance of the support of mass density for Navier-Stokes-Boltzmann equations induced by the proof of the corresponding result in  Xin \cite{zx} for Navier-Stokes equations.

\begin{theorem}[\textbf{Invariance of the support of mass density}]\label{th:3.1}\ \\
Let $T>0$ and $\supp\rho_0(x)\subseteq B_{R_0}$ and
\begin{equation}\label{eq:2.10}
  \mu > 0, \quad \lambda+\frac{2}{d}\mu > 0,\quad k(\theta)\geq 0.\end{equation}
Then for any smooth solution
$$I(t,x,v,\Omega)\in L^2(\mathbb{R}^+ \times S^{d-1};C^1([0,T);H^s(\mathbb{R}^d))),\ (\rho,u,S)(t,x)\in C^1([0,T);H^s(\mathbb{R}^d))$$
to the Cauchy problem (\ref{eq:2.1})--(\ref{eq:2.2}) satisfying (\ref{eq:2u})-(\ref{eq:2.2*}), the support of the density $\rho(t,x)$ will not grow in time. More precisely, it holds that
\begin{equation}\label{eq:2.11}
D(t)=2R(t)=2L=2\widetilde{L}=2R_0, \qquad \forall t \in [0,T),
\end {equation}
i.e., the estimates (\ref{cd:1}) and (\ref{cd:11}) hold with $\alpha=0$ and $L$, $\widetilde{L}$ given by (\ref{eq:2.11}).

\end{theorem}

\begin{proof}
Due to $\supp\rho_0(x)\subseteq B_{R_0}$ and the hyperbolic property of  continuity equation, we know that $R(t)$ is a well-defined finite positive number for any $t \ge 0$.

We denote by $ S_p(t) $ the compact domain that is the image of $ \supp\rho_0(x)$ under the flow map, i.e.,
\begin{equation}\label{zhi}
S_p(t)=\{x|x=x(t;\xi_0), \quad \forall \xi_0 \in \supp\rho_0(x)\},
\end{equation}
where $ x(t; \xi_0)$ is the particle path starting from $\xi_0$ when $t=0$, namely,
\begin{equation}\label{particle}  \frac{d}{\text{d}t}x(t; \xi_0)=u(t,x(t; \xi_0)),   \quad x(0;\xi_0)= \xi_0.
\end{equation}
      It follows from the continuity equation that the smooth solution is simply supported along the particle paths, so
       $$ \supp_x \rho(t,x)=  S_p(t).$$
In the vacuum domain, $K_a(t,x,v,\rho,\theta)=0$  (see \cite{gp}), then from system (\ref{eq:2.1}) we have
\begin{equation}\label{100}
\nabla\cdot \mathbb{T}=0,\
\nabla\cdot (u \mathbb{T})=0, \quad \text{in  } (S_p(t))^c.
\end{equation}
 According to the proof of Lemma \ref{lemma:2.4}, we have
 \begin{equation}\label{eq:4.2}
 \begin{cases}
 \partial_iu_i(t,x)\equiv 0 ,\\
 \partial_iu_j(t,x)\equiv -\partial_ju_i(t,x) \quad(i\neq j)
 \end{cases}
 \end{equation}
 in $(S_p(t))^c$. Direct calculations lead to
 \begin{equation}\label{ry}
 \begin{split}
 \partial^2_{ij}u_k
 &=\partial_i(\partial_ju_k)=-\partial_i(\partial_ku_j)=-\partial^2_{ik}u_j\\
 &=\partial_j(\partial_iu_k)=-\partial_j(\partial_ku_i)=-\partial_j(\partial_ku_i)=\partial^2_{ik}u_j.
 \end{split}
\end{equation}
Thus,
\begin{equation}\label{ry1}
\partial^2_{ij}u_k=0, \ 1\leq i,j,k \leq d, \quad \text{in}  \  (S_p(t))^c.
\end{equation}

Since $ u(t,\cdot)\in H^s(\mathbb{R}^d)$, then
\begin{equation*}
u(t,x)\equiv 0 ,\quad \text{in} \    (S_p(t))^c.
\end{equation*}
From the definition of $S_p(t)$, we have $ u(t,x(t;x_0))\equiv 0$, for $x_0 \in \partial\supp \rho_0(x)$, where $\partial\supp \rho_0(x)$ is the boundary of  $\supp \rho_0(x)$.
Then
\begin{equation*}
x(t;\xi_0)\equiv \xi_0,  \quad \xi_0 \in \partial \supp\rho_0(x),\quad \text{i.e.},\  S_p(t)= \supp\rho_0(x).
\end{equation*}
\end{proof}

\subsection{\textbf{Navier-Stokes-Boltzmann equations without heat conduction}}\ \\
 As an immediate consequence of Theorem \ref{th:2.1}, Theorem \ref{th:3.1} and the second law of thermodynamics, we have the following blow-up results on the smooth solutions to the Cauchy problem (\ref{eq:2.1})--(\ref{eq:2.2}).

\begin{theorem}\label{co:2.9}
Let $T> 0$ and $\supp\rho_0(x)\subseteq B_{R_0}$. Consider the viscous compressible flows in radiation hydrodynamics without heat conduction, i.e.,
\begin{equation}\label{eq:2.13}
  \mu > 0, \quad\gamma+\frac{2}{d}\mu > 0,\quad k(\theta)= 0.
  \end{equation}
Then any smooth solution
$$ I(t,x,v,\Omega)\in L^2(\mathbb{R}^+ \times S^{d-1};C^1([0,T);H^s(\mathbb{R}^d))),\ (\rho,u,S)(t,x)\in C^1([0,T);H^s(\mathbb{R}^d))$$
 of the Cauchy problem (\ref{eq:2.1})--(\ref{eq:2.2}) satisfying (\ref{eq:2u})-(\ref{eq:2.2rr}) will blow up in finite time. More precisely, the life span is estimated in (\ref{eq:2.9}) with $\alpha=0$ and $L$ given by (\ref{eq:2.11}), $\underline{S}$ is to be given.
\end{theorem}

\begin{proof}
By Theorem \ref{th:3.1}, we have $\supp_x \rho(t,x)\subseteq B_{R_0}$, so we need only verify (\ref{cd:2}).


$\forall y \in S_p(t)$, $\exists \ y_0$, such that $y=y(t;y_0 )$ and $ y(t;y_0)$ is the photon path starting from $y_0$ when $t=0$, \text{i.e.},
      $$   \frac{d}{\partial{t}}y(t;y_0)=c\Omega ,     \qquad  y(0;y_0)= y_0 .                   $$
We have
\begin{equation}\label{erer}
|y-y_0|=|c\Omega t|> 2R_0, \ \text{for} \ t >T_c=\frac{2R_0}{c},
\end{equation}
and thus
\begin{equation}\label{erer1}
|y_0|> R_0,\ \text{i.e.},\ y_0\in (S_p(0))^c.
\end{equation}
According to Lemma \ref{lemma:2.1}, we have
\begin{equation} \label{eq:*}
(I-\overline{B}(v))(t,y(t;y_0 ))=(I_0-\overline{B}(v))(y_0 )\exp\Big(\int_0^t -cK_a\cdot(\tau,y(\tau;y_0 ),v,\rho,\theta)\text{d}\tau\Big).   \end{equation}
From (\ref{eq:*}) and (\ref{eq:2.2rr}), we have
\begin{equation}\label{erer2}
(I-\overline{B}(v))(t,y(t;y_0 ))=0, \ t >T_c=\frac{2R_0}{c}.
\end{equation}
That is to say,
\begin{equation}\label{gh1}
I(t,x,v,\Omega)\equiv \overline{B}(v), \ \forall x\in \supp\rho_0(x),\ t >T_c=\frac{2R_0}{c}.
\end{equation}

Let $\xi_0\in \supp\rho_0(x)\subseteq B_{R_0}$ and $ x(t; \xi_0)$ be the particle path defined by (\ref{particle}). According to the continuity equation and Theorem \ref{th:3.1}, we have $x(t; \xi_0)\in \rho_0(x)\subseteq B_{R_0}$ for $t\in [0,T)$. Along $ x(t; \xi_0)$, we obtain
\begin{equation}\label{en}
\begin{split}
\frac{d}{\text{d}t}S(t,x(t;\xi_0))p_m&=(\partial_tS+u(t,x(t;\xi_0))\cdot \nabla S)p_m\\
&=(\gamma-1)\int_0^\infty \int_{S^{d-1}} \left(1-\frac{u\cdot \Omega}{c}\right)K_a\cdot(I-\overline{B}(v)) \text{d}\Omega \text{d}v\\
&\quad+(\gamma-1)\Big(\nabla \cdot (u \mathbb{T})-u\cdot (\nabla \cdot \mathbb{T})\Big).
\end{split}
\end{equation}
If $ T\leq T_c$, the proof is finished.\\
If $ T > T_c $, according to (\ref{gh1}), (\ref{en})  and Lemma \ref{lemma:2.2}, we have
$$\frac{d}{\partial{t}}S(t,x(t;\xi_0)) \geq 0,\quad \text{for}\quad t\in (T_c,T). $$
Then
\begin{equation}\label{7788}
S(t,x) \geq \min_{x\in B_{R_0}} S(T_c,x)=S_1, \quad \text{for} \quad  (t,x)\in (T_c,T)\times B_{R_0}.
\end{equation}
Since $ S(t,x)\in C^1([0,T_c];H^s(\mathbb{R}^d)), \ s\geq \frac{d}{2}+2$,
from Sobolev's imbedding theorem, we know that there exists a positive constant $S_2$ such that
$$  |S(t,x)|\leq S_2, \quad \text{for}  \quad (t,x)\in [0,T_c]\times B_{R_0}.$$
Therefore,
$$ S(t,x)\geq \min\{S_1,-S_2\}=\underline{S}, \quad \text{for} \quad (t,x)\in [0,T)\times B_{R_0}.$$
\end{proof}

\begin{remark}\label{r31}
Compared with Navier-Stokes equations, we obtained that $\frac{d}{\partial{t}}S(t,x(t;\xi_0)) \geq0$  along the particle path for $t\geq T_c$ instead of for all $t\geq 0$, which is caused by the viscosity effect along with the radiation effect. This is different from the Euler equations, where the entropy is  invariant along the particle path. That is to say, the viscosity and radiation can change the mechanical properties of the fluid in some sense.
\end{remark}

Now we consider the smooth solution in a broader class of functions that do not need higher regularity. We first introduce the well known Reynolds transport theorem (\cite{kong}).
\begin{lemma}\label{3.1}Let $S_{p}(t)$ be defined by (\ref{zhi}).
Then for any $Q(t,x)\in C^1(\mathbb{R}^+ \times \mathbb{R}^d)$, we have
$$
\frac{d}{\text{d}t}\int_{S_p(t)}  Q(t,x)\text{d}x= \int_{S_p(t)}  \partial_t Q(t,x)\text{d}x+\int_{\partial {S_p(t)}} Q(t,x)(u(t,x)\cdot n)dS,
$$
where $n$ is the outward unit normal vector of $\partial {S_p(t)}$.
\end{lemma}
The proof is a direct calculation, here we omit it.

From Reynolds transport theorem and the continuity equation, we have
\begin{equation}\label{en1}
\frac{d}{\text{d}t}\int_{S_p(t)}  \rho(t,x)\text{d}x=\int_{S_p(t)}\partial_t\rho(t,x)\text{d}x+\int_{\partial S_{p}(t)} \rho(t,x)(u(t,x)\cdot n)dS=0,
\end{equation}
which implies
\begin{equation}\label{ff}
m(t)=\int_{\mathbb{R}^d} \rho(t,x)\text{d}x=\int_{S_p(t)}  \rho(t,x)\text{d}x=\int_{S_p(0)}  \rho_0(x)\text{d}x=m(0).
\end{equation}

Now we  give the following blow-up result.

\begin{theorem}\label{co:2.2} Let $T>0$, $\supp\rho_0(x)\subseteq B_{R_0}$ and
\begin{equation}
\label{eq:1.5jkk}
  \mu > 0, \quad \lambda+\frac{2}{d}\mu > 0,\ k(\theta)= 0.\end{equation}
Let $(I,\rho,u,S)$ be the smooth solution of the Cauchy problem (\ref{eq:2.1})--(\ref{eq:2.2}) satisfying (\ref{eq:2u})-(\ref{eq:2.2rr}) and
\begin{equation*}
\begin{cases}
I(t,x,v,\Omega)\in L^2(\mathbb{R}^+ \times S^{d-1};C^1([0,T)\times \mathbb{R}^d))), \ (\rho,S)(t,x)\in C^1([0,T)\times\mathbb{R}^d);\\[8pt]
u(t,x)\in C^1([0,T)\times\mathbb{R}^d),\ \partial^2_xu(t,x)\in C([0,T)\times\mathbb{R}^d).
\end{cases}
\end{equation*}
If
\begin{equation}\label{eq:2.16}
\| u(t,x)\|_{L^\infty} < C(t+\kappa)^{-\beta}
\end{equation}
with $C$ and $\beta<1 $ being positive constants and $1< \gamma < 1+\frac{1}{d(1-\beta)}$,
then
$$ T < +\infty. $$
\end {theorem}

\begin{proof}Set $\alpha=1-\beta$, then $0<\alpha < 1$. Consider the particle path $x(t;\xi_0)$ with $\xi_0 \in \partial{B_{R_0}}$. Then
\begin{equation}\label{t1}
\begin{split}
|x(t; \xi_0)| &\leq |\xi_0 |+\int_{0}^{t} |u(s,x(s;\xi_0))|ds\\
& \leq R_0+C\int_{0}^{t} (s+\kappa)^{-\beta}ds=R_0+\frac{C}{1-\beta}((t+\kappa)^{1-\beta}-\kappa^{1-\beta})\\
&\leq \max\Big\{R_0,\frac{C}{1-\beta}\Big\}(t+\kappa)^\alpha=L(t+\kappa)^\alpha.
\end{split}
\end{equation}
The condition (\ref{cd:1}) holds with $\alpha=1-\beta$ and  $L$ depending only on $R_0, C, \beta$.

$\forall y \in S_p(t)$, $\exists y_0$, such that $y=y(t;y_0 )$, where $ y(t;y_0)$ is the photon path starting from $y_0$ when $t=0$.
We have
\begin{equation}\label{erer}
|y-y_0|=|c\Omega t|> R_0+L(t+\kappa)^\alpha, \quad t >\widetilde{T}_c=\kappa+\frac{R_0}{c}+\big(\frac{2^\alpha L}{c}\big)^{\frac{1}{1-\alpha}},
\end{equation}
and thus
\begin{equation}\label{erer1}
|y_0|> R_0,\ \text{i.e.},\ y_0\in (S_p(0))^c.
\end{equation}
From (\ref{eq:*}), we have
\begin{equation}\label{erer2}
(I-\overline{B}(v))(t,y(t;y_0 ))=0, \ t >\widetilde{T}_c.
\end{equation}
That is to say,
\begin{equation}\label{gh1}
I(t,x,v,\Omega)\equiv \overline{B}(v), \ \forall x\in S_{p}(t),\ t >\widetilde{T}_c.
\end{equation}
If $ T\leq \widetilde{T}_c$, then the proof is finished.\\
If $ T > \widetilde{T}_c $, let $\xi_0 \in \supp \rho_0(x)\subseteq B_{R_0}$, $ x(t; \xi_0)$ is the particle path defined by (\ref{particle}). According to the continuity equation, we have $x(t; \xi_0)\in S_p(t) $ for $t\in [0,T)$. From (\ref{en}), (\ref{gh1}) and Lemma \ref{lemma:2.2}, we have
$$\frac{d}{\partial{t}}S(t,x(t;\xi_0)) \geq 0. $$
Then it is similar to the proof of Theorem \ref{co:2.9} that there exists a constant $\underline{S}$ such that
\begin{equation}\label{t2}
 S(t,x)\geq \underline{S}, \quad  \forall x\in S_p(t).
\end{equation}
Replacing in (\ref{def}) the integration domain of $x$ with $S_{p}(t)$, we get
\begin{equation}\label{def2}
\begin{split}
\widetilde{I}_r(t)
=&\int_{S_p(t)} |x-(t+\kappa)u|^{2}\rho \text{d}x+\frac{2}{\gamma-1}(t+\kappa)^{2}\int_{S_p(t)}p_m \text{d}x+(t+\kappa)\widetilde{Q}_r(t),
\end{split}
\end{equation}
where
$$
\widetilde{Q}_r(t)=\int_{S_p(t)}\int_0^\infty  \int_{S^{d-1}} \frac{2c(t+\kappa)-2x\cdot\Omega}{c^2}(I-\overline{B}(v))\text{d}\Omega \text{d}v \text{d}x.
$$
Then according to Lemma \ref{3.1}, (\ref{gh1}) and the proof of Proposition \ref{pro:2.1},  we have
\begin{equation}\label{cd:4q}
\begin{split}
\frac{d}{\text{d}t}\widetilde{I}_r(t)=&\frac{2}{\gamma-1}(2-d(\gamma-1))(t+\kappa)\int_{S_p(t)}p_m \text{d}x\\
&+2(t+\kappa)^2\int_{S_p(t)}\nabla\cdot(u\mathbb{T})\text{d}x+\int_{S_p(t)}x\cdot (\nabla\cdot \mathbb{T})\text{d}x,
\end{split}
\end{equation}
for $ t >\widetilde{T}_c$.
Comparing with the proofs of Proposition \ref{pro:2.1} and Theorem \ref{th:2.1}, we need only prove that
\begin{equation}\label{zheng1}
2(t+\kappa)^2\int_{S_p(t)}\nabla\cdot(u\mathbb{T})\text{d}x+\int_{S_p(t)}x\cdot (\nabla\cdot \mathbb{T}) \text{d}x=0.
\end{equation}
According to the proof of Theorem \ref{th:3.1}, we know that
\begin{equation}\label{zheng}
\partial^2_{ij}u_k=0, \ 1\leq i,j,k \leq d, \quad \text{in}  \  (S_p(t))^c.
\end{equation}
Then there exists a matrix $\mathbb{N}(t)$ and a vector $b(t)$ such that
\begin{equation}\label{ry2}
u=\mathbb{N}(t)x+b(t), \ \forall \ x\in  (S_p(t))^c.
\end{equation}
Due to (\ref{eq:4.2}), we have
$$
\mathbb{N}(t)+\mathbb{N}^\top(t)=0, \ \forall t\in [0,T),$$
i.e., $\mathbb{N}(t)$ ia a antisymmetric matrix for any $t\in[0,T)$, then
$$
\mathbb{T}=0,\ \forall \ x \in (S_p(t))^c.
$$
By direct calculation, we have
\begin{equation}\label{t3}
\begin{split}
&   \int_{S_p(t)}x\cdot (\nabla\cdot \mathbb{T})\text{d}x=- \int_{S_p(t)} \text{tr}(\mathbb{T})\text{d}x \\ &=-(2\mu+d\lambda) \int_{S_p(t)}  \nabla\cdot u(t,x) \text{d}x
=-(2\mu+d\lambda) \int_{\partial S_p(t)}   u(t,x)\cdot n ds\\
&= -(2\mu+d\lambda) \int_{\partial S_p(t)}   (\mathbb{N}(t)x+b(t))\cdot n ds                          \\
& =-(2\mu+d\lambda) \int_{ S_p(t)}  \nabla \cdot (\mathbb{N}(t)x+b(t)) \text{d}x =-(2\mu+d\lambda) \int_{ S_p(t)}  \text{tr}(\mathbb{N}(t)) \text{d}x=0.
\end{split}
\end{equation}
and
\begin{equation}\label{200}
\int_{S_p(t)}\nabla\cdot(u\mathbb{T})\text{d}x=\int_{\partial S_p(t)} (u\mathbb{T})\cdot ndx=0.
\end{equation}

Combining (\ref{t1}), (\ref{t2}), and (\ref{t3})-(\ref{200}), according to  the proof of Theorem \ref{th:2.1},
we know that the life span of this smooth solution is finite for $1 < \gamma < 1+\frac{1}{d(1-\beta)} $.
\end{proof}
\begin{remark}\label{shuoming}
Theorem \ref{co:2.2} is also true for Euler-Boltzmann equations, since (\ref{zheng1}) is naturally valid if $\mu=\lambda=k(\theta)=0$.
\end{remark}
\subsection{\textbf{Navier-Stokes-Boltzmann equations with heat conduction}}\ \\
Compared with Theorem \ref{co:2.9}, we have the similar conclusion for viscous flow with heat conduction in radiation hydrodynamics. It is  an immediate consequence of Theorem \ref{th:2.0000} and  Theorem \ref{th:3.1}.
\begin{theorem}\label{co:2.1}
Let $T>0$ and $\supp\rho_0(x)\subseteq B_{R_0}$. Consider the viscous compressible flows in radiation hydrodynamics with heat conduction, i.e.,
\begin{equation}\label{eq:2.132}
  \mu > 0, \quad\gamma+\frac{2}{d}\mu > 0,\quad k(\theta) > 0.
  \end{equation}
Then any smooth solution
$$ I(t,x,v,\Omega)\in L^2(\mathbb{R}^+ \times S^{d-1};C^1([0,T);H^s(\mathbb{R}^d))),\ (\rho,u,S)(t,x)\in C^1([0,T);H^s(\mathbb{R}^d))$$
 of the Cauchy problem (\ref{eq:2.1})--(\ref{eq:2.2}) satisfying (\ref{eq:2u})-(\ref{eq:2.2*}) and (\ref{bbbb}) will blow up in finite time.
\end{theorem}

\begin{remark}\label{r44}
As to the  isentropic compressible Navier-Stokes-Boltzmann equations in multi-dimensional case, from the transport equation of photons, continuity equation, momentum equations and  the physical relation for polytropic gas
$$ E_m=\frac{1}{2}\rho u^2+\frac{p_m}{\gamma-1},$$
we can get the energy equation for the isentropic flow
\begin{equation} \label{eq:5.00}
\begin{split}
\partial_t(E_m+E_r)+
\nabla\cdot((E_m+p_m)u+F_r)
=u\cdot(\nabla \cdot  \mathbb{T}).
\end{split}
\end{equation}
The method used in Theorem 3.1 to get the invariance of the support of density fails in this case, because we can only  get
$$ u\cdot (\nabla \cdot\mathbb{ T})=0 $$
in the vacuum domain, while for the non-isentropic flow, we have
\begin{equation*}
\nabla\cdot \mathbb{T}=0,\
\nabla\cdot (u \mathbb{T})=0
\end{equation*}
in the vacuum domain. According to the proof of Theorem \ref{th:3.1}, we know that the  conclusions for the non-isentropic flow cannot go directly to the isentropic flow.
\end{remark}

\section{Multi-dimensional isentropic flow with degenerate viscosity coefficients}
Through the discussion in Section 3, we can see that the results about Navier-Stokes-Boltzmann equations cannot be directly extended to the isentropic case with only compactly supported mass density in the case that the viscosity coefficients are constants, so we consider the situation when the  viscosity coefficients depend on mass density, which is motivated by the physical consideration that in the derivation of the Navier-Stokes equations from the Boltzmann equations through the Chapman-Enskog expansion to the second order, cf.\cite{tlt}, the viscosity coefficients are not constant but depend on temperature. For isentropic flow, this dependence is reduced to the dependence on the density by the Boyle and Gay-Lusac law for ideal gases. We consider the system
\begin{equation}
\begin{cases}
\label{eq:aaa2}
\displaystyle
\frac{1}{c}\partial_tI+\Omega\cdot\nabla I=-K_a\cdot(I-\overline{B}(v)),\\[10pt]
\displaystyle
\partial_{t}\rho+\nabla\cdot(\rho u)=0,\\[10pt]
\displaystyle
\partial_{t}(\rho u)+\nabla\cdot(\rho u \otimes u)
+\nabla p_m =
\frac{1}{c}\int_0^\infty \int_{S^{d-1}} K_a\cdot(I-\overline{B}(v))\Omega \text{d}\Omega \text{d}v+\nabla \cdot \mathbb{T},
 \end{cases}
\end{equation}
where $ x \in \mathbb{R}^d $, $ d \geq 2 $.
\begin{equation} \label{fg}
K_a(t,x,v,\rho)=o(\rho)=\rho \widetilde{K}_a(t,x,v,\rho),
\end{equation}
where $\lim_{\rho\rightarrow 0}\widetilde{K}_a(t,x,v,\rho)=0$.
We consider only polytropic gas, i.e., $ p_m(\rho)= \rho^\gamma=R \rho \theta $.
$\mathbb{T}$ is the stress tensor given by
\begin{equation}
\label{eq:bbb}
\mathbb{T}=\mu(\rho)(\nabla u+(\nabla u)^\top)+\lambda(\rho)(\nabla\cdot u)\mathbb{I}_d,
\end{equation}
where $\mu(\rho)=n_1\rho^\delta, \ \lambda(\rho)=n_2\rho^\delta$ and  $\frac{2}{d}\mu(\rho)+\lambda(\rho)=\rho^\delta$, with $ 1< \delta \leq \gamma$. Here $n_1>0$ and $n_2$ are constants. We first give the definition of  regular solutions of system (\ref{eq:aaa2}).

\begin{remark}
The condition (\ref{fg}) for the isentropic flow can be satisfied when the absorption coefficient is given by, for example (see \cite{gp} or \cite{sjx}),
$$
K_a(t,x,v,\rho)=C\rho \theta^{-\frac{1}{2}}\exp\Big(-\frac{C}{\theta^{\frac{1}{2}}}\Big(\frac{v-v_0}{v_0}\Big)^2\Big),
$$
where $C$ is a positive constant, $v_0$ is the fixed frequency. Then we have
$$
\lim_{\rho\rightarrow 0}\frac{K_a(t,x,v,\rho)}{\rho}=\lim_{\theta \rightarrow 0} \theta^{-\frac{1}{2}}\exp\Big(-\frac{C}{\theta^{\frac{1}{2}}}\Big(\frac{v-v_0}{v_0}\Big)^2\Big)=0.
$$
\end{remark}
\begin{definition}[\text{\textbf{Regular solution}}]\ \\
A solution $(I(t,x,v,\Omega),\rho(t,x),u(t,x))$ of problem (\ref{eq:aaa2}) is called a regular solution in $ [0,T)\times \mathbb{R}^d\times \mathbb{R}^+ \times S^{d-1}$ if \\[6pt]
\text{(i)} $ I(t,x,v,\Omega)\in L^2(\mathbb{R}^+ \times S^{d-1};C^1([0,T)\times\mathbb{R}^d))$, $\rho(t,x)\in C^1([0,T) \times \mathbb{R}^d)$,  $\rho\geq 0$, and $ u(t,x)\in C^1([0,T)\times \mathbb{R}^d))$ , $\partial^2_{x_ix_j}u(t,x)\in C([0,T) \times \mathbb{R}^d)$, $i, j=1,2,...,d$,\\[6pt]
\text{(ii)} $\rho^{\frac{\delta-1}{2}}(t,x) \in C^1([0,T) \times \mathbb{R}^d)$ with $1<\delta \leq \gamma$.
\end{definition}
We assign the initial condition
\begin{equation} \label{eq:2.2e}
I|_{t=0}=I_0(x,v,\Omega), \quad (\rho,u)|_{t=0}=(\rho_0(x),u_0(x))
\end{equation}
satisfying
\begin{equation} \label{eq:df}
\begin{split}
&I_0 \geq \overline{B}(v), \ \text{for} \ (x,v,\Omega)\in \mathbb{R}^{d}\times \mathbb{R}^{+}\times S^{d-1}; \  I_0\equiv \overline{B}(v), \ \text{for} \ |x|\geq R_0,\\
&\sup \rho_0(x)\subseteq B_{R_0},\ \supp u_0(x)\subseteq B_{R_0}.
\end{split}
\end{equation}

\begin{lemma}[\text{\textbf{Finite expansion of the vacuum domain}}]
\label{lemma:3.1}\ \\
Let $T>0$. Suppose that
$(I,\rho,u)$
is a regular solution to the Cauchy problem (\ref{eq:aaa2})-(\ref{eq:2.2e}) satisfying (\ref{eq:df}). We denote by $ x(t; \xi_0)$ the particle path starting from $\xi_0$ when $t=0$, then we have
$$
x(t; \xi_0)=\xi_0, \quad \text{for}\quad \xi_0\in \partial{B_{R_0}}, \quad  t\in [0,T) .
$$
Moreover, $\supp_x \rho(t,x)=\supp_x u(t,x)\subseteq B_{R_0}$  \text{for} $t\in [0,T)$.
\end{lemma}

\begin{proof}
We introduce $w=\rho^{\frac{\delta-1}{2}}$, which is first induced by Yang-Zhu \cite{tyc2}. Then  system (\ref{eq:aaa2})  can be written as
\begin{equation}
\begin{cases}
\label{eq:ccc}
\displaystyle
\frac{1}{c}\partial_tI+\Omega\cdot\nabla I=-K_a\cdot(I-\overline{B}(v)),\\[8pt]
\partial_{t}w+u\cdot \nabla w+\frac{\delta-1}{2}w \nabla \cdot u=0,\\[8pt]
\displaystyle
\partial_tu+u\cdot\nabla u +\frac{2\gamma}{\delta-1}w^{\frac{2(r-\delta)}{\delta-1}}w\nabla w=g(w,u)+\frac{1}{c}\int_0^\infty \int_{S^{d-1}} \widetilde{K}_a\cdot(I-\overline{B}(v))\Omega \text{d}\Omega \text{d}v,
 \end{cases}
\end{equation}
where
\begin{equation} \label{eq:5.2}
\begin{split}
g(w,u)=&\frac{2\delta}{\delta-1}w\nabla w \cdot \left(n_1\nabla u+n_1(\nabla u)^\top+n_2 \nabla\cdot u \mathbb{I}_d\right)\\
&+w^2\left(\nabla \cdot (n_1\nabla u+n_1(\nabla u)^\top)+n_2\nabla (\nabla\cdot u)\right).
\end{split}
\end{equation}
Combining (\ref{fg}), (\ref{eq:ccc}) and (\ref{eq:5.2}), we have
\begin{equation}\label{eq:5.3}
\partial_tu+u\cdot\nabla u=0, \quad \text{in} \quad (S_p(t))^c.
\end{equation}
That is to say, $u$ is invariant along the particle path on $ (S_p(t))^c$. Thus,  from (\ref{eq:df}) we have
\begin{equation}\label{eq:5.4}
u(t,x)\equiv 0, \quad \text{in} \quad  (S_p(t))^c.
\end{equation}
Using the $C^1$ continuity of $u(t,x)$, we get
\begin{equation*}
\frac{d}{\text{d}t}x(t;\xi_0)=u(t;x(t; \xi_0))\equiv 0 ,  \quad \xi_0 \in \partial\supp \rho_0(x),
\end{equation*}
and thus $x(t; \xi_0)\equiv \xi_0$.
So
\begin{equation} \label{eq:5.5} \supp_x \rho(t,x)=\supp_x u(t,x)\subseteq B_{R_0},\ \text{for}\ t\in [0,T).\end{equation}
\end{proof}
Now we introduce an important lemma.
 We assume that $(I,\rho,u)$ is a regular solution of the Cauchy problem (\ref{eq:aaa2}) and (\ref{eq:2.2e}) satisfying (\ref{eq:df}).
Then, we have the following property for $I_r(t)$ defined by (\ref{def}).
\begin{lemma}
\label{lemma:3.3}
Let $1<\delta\leq\gamma$ be a positive constant.  If
\begin{equation}\label{co:22}
\frac{d}{\text{d}t}I_r(t) \leq \frac{2-d(\gamma-1)}{t+\kappa}I_r(t)+\frac{\delta(\gamma-1)}{2\gamma(t+\kappa)^2}I_r(t)+\frac{\gamma-\delta}{\gamma}|B_{R_0}|,
\end{equation}
then we have
$$
I_r(t)\leq C(t^{2-d(\gamma-1)}+t\ln t+1),\ t\in [0,T),
$$
where $C$ is a generic positive constant. In particular, the above estimate implies that
$$ T < +\infty. $$
\end{lemma}

\begin{proof} Solving (\ref{co:22}) directly, we get
\begin{equation}\label{k1}
\begin{split}
I_r(t) \leq & \Big(\frac{t+\kappa}{\kappa}\Big)^{2-d(\gamma-1)}e^{-\frac{\delta(\gamma-1)}{2\gamma(t+\kappa)}} \Big( e^{\frac{\delta(\gamma-1)}{2\gamma\kappa}}I_r(0) \\
&+\frac{\gamma -\delta}{\gamma }|B_{R_0}|  \int_{0}^{t}\Big(\frac{\tau+\kappa}{\kappa}\Big)^{d(\gamma-1)-2} e^{ \frac{\delta(\gamma-1)}{2\gamma(\tau+\kappa)}} d \tau \Big).
\end{split}
\end{equation}
For $d(\gamma-1)-2 \neq -1$, from (\ref{k1}) we get,
\begin{equation}\label{eq:3.ode1}
\begin{split}
I_r(t)
\leq &\Big(\frac{t+\kappa}{\kappa}\Big)^{2-d(\gamma-1)}e^{-\frac{\delta(\gamma-1)}{2\gamma(t+\kappa)}}
\left(e^{\frac{\delta(\gamma-1)}{2\gamma\kappa}}I_r(0)-\frac{(\gamma -\delta)\kappa}{\gamma(d(\gamma-1)-1)}|B_{R_0}|e^{\frac{\delta(\gamma-1)}{2\gamma\kappa}} \right) \\[6pt]
&+\frac{(\gamma-\delta)(t+\kappa)}{\gamma (d(\gamma-1)-1)}|B_{R_0}| e^{-\frac{\delta(\gamma-1)}{2\gamma(t+\kappa)}} e^{\frac{\delta(\gamma-1)}{2\gamma\kappa}}.
\end{split}
\end{equation}
For $d(\gamma-1)-2=-1$, from (\ref{k1}) we get
\begin{equation}\label{k2}
\begin{split}
I_r(t) \leq & \Big(\frac{t+\kappa}{\kappa}\Big)^{2-d(\gamma-1)}e^{-\frac{\delta(\gamma-1)}{2\gamma(t+\kappa)}} \Big( e^{\frac{\delta(\gamma-1)}{2\gamma\kappa}}I_r(0)+\frac{\gamma -\delta}{\gamma}\kappa|B_{R_0}|   e^{ \frac{\delta(\gamma-1)}{2\gamma}} (\ln(t+\kappa)-\ln\kappa) \Big).
\end{split}
\end{equation}
We construct a function
$$
g(x)=e^{-\frac{\delta(\gamma-1)}{2\gamma}x}-(x+1),  \quad \text{for} \ x\in [0,1].
$$
Direct calculation leads to
$$
g'(x)=\left(-\frac{\delta(\gamma-1)}{2\gamma}\right)e^{-\frac{\delta(\gamma-1)}{2\gamma}x}-1<0 ,  \quad \text{for} \ x\in [0,1].
$$
From $g(0)=0$, we have
\begin{equation}\label{co:24}
e^{-\frac{\delta(\gamma-1)}{2\gamma}x}\leq x+1,  \quad \text{for} \ x\in [0,1],
\end{equation}
i.e.,
\begin{equation}\label{co:25}
e^{-\frac{\delta(\gamma-1)}{2\gamma(t+\kappa)}} \leq \frac{1}{t+\kappa}+1.
\end{equation}
Then from (\ref{eq:3.ode1}), (\ref{k2}) and (\ref{co:25}), for $d(\gamma-1)-2 \neq -1$, we have
\begin{equation}\label{eq:32}
I_r(t) \leq C(t^{2-d(\gamma-1)}+ t+1), \ \text{for} \ t\in [0,T);
\end{equation}
for $d(\gamma-1)-2=-1$, we have
\begin{equation}\label{eq:3g2}
I_r(t) \leq C(t+(t+\kappa)\ln(t+\kappa)+1), \ \text{for} \ t\in [0,T).
\end{equation}
 Due to the definition of $I_r(t)$ and Jensen's inequality, we also get
\begin{equation}\label{ode:3.21}
\begin{split}
I_r(t)&\geq \frac{2(t+\kappa)^2}{\gamma-1}\int_{B_{R_0}}p_m(t,x)\text{d}x \geq  \frac{2(t+\kappa)^2}{\gamma-1}  |B_{R_0}|\int_{B_{R_0}} \rho^\gamma(t,x)\frac{\text{d}x}{|B_{R_0}|} \\
&\geq  \frac{2(t+\kappa)^2}{\gamma-1}  |B_{R_0}|^{1-\gamma} m(0)^\gamma,
\end{split}
\end{equation}
where we have used the fact that
$$ \int_{B_{R_0}} \rho(t,x) \text{d}x =\int_{B_{R_0}}\rho_0(x)\text{d}x=m(0).$$
Combining (\ref{eq:32}) and (\ref{ode:3.21}), for $d(\gamma-1)-2\neq -1$, we have
\begin{equation*}
\frac{2(t+\kappa)^2}{\gamma-1}  |B_{R_0}|^{1-\gamma} m(0)^\gamma \leq C(t^{2-d(\gamma-1)}+ t+1);
\end{equation*}
for $d(\gamma-1)-2 = -1$, we have
\begin{equation*}
\frac{2(t+\kappa)^2}{\gamma-1}  |B_{R_0}|^{1-\gamma} m(0)^\gamma \leq C(t+(t+\kappa)\ln(t+\kappa)+1), \ \text{for} \ t\in [0,T),
\end{equation*}
which imply that $T< +\infty$.
\end{proof}

Now we prove the blow-up of regular solutions for multi-dimensional isentropic flow.

\begin{theorem}[\textbf{Multi-dimensional isentropic flow}]\label{th:2.3}\ \\
Let $T>0$. Suppose that
 $ (I(t,x,v,\Omega),\rho(t,x),u(t,x))$
 is a regular solution to the Cauchy problem (\ref{eq:aaa2}) and (\ref{eq:2.2e}) satisfying (\ref{eq:df}).
Then
\begin{equation}\label{eq:bbb2}
T < +\infty.
\end{equation}
\end{theorem}

\begin{proof}
From the continuity equation, momentum equations and integrating by parts, we have
\begin{equation} \label{eq:5.6}
\begin{split}
\frac{d}{\text{d}t}I_r(t)=&\frac{2}{\gamma-1}(2-d(\gamma-1))(t+\kappa)\int_{\mathbb{R}^d} p_m \text{d}x+Q_r(t)+J_1+J_2,
\end{split}
\end{equation}
which has two additional terms $J_1$ and $J_2$ compared with (\ref{cd:4}) for non-isentropic case and
\begin{equation*}
\begin{split}
J_1=&-2(t+\kappa)\int_{\mathbb{R}^d} x \cdot (\nabla \cdot \mathbb{ T}) \text{d}x,\\
J_2=&2(t+\kappa)^2 \int_{\mathbb{R}^d}\partial_t\left(\frac{1}{2}\rho|u|^2+\frac{p_m}{\gamma-1}+\widetilde{E}_r(t)\right)\text{d}x.
\end{split}
\end{equation*}
We first look at $J_1$. Since
\begin{equation*}
\begin{split}
\nabla \cdot (x\cdot \mathbb{T })
&=\sum_{i=1}^{d}\sum_{j=1}^{d}\left( x_{i}\frac{\partial{\mathbb{T }(ij)}}{\partial{x_j}}+\delta_{ij}\mathbb{T }(ij)\right)=x\cdot (\nabla \cdot \mathbb{T }) +\sum_{i=1}^{d}\mathbb{T }(ii)\\
&=x\cdot (\nabla \cdot \mathbb{T})+\Big(\lambda(\rho)+\frac{2}{d}\mu(\rho)\Big)\nabla \cdot u(t,x)\\
&=x\cdot (\nabla \cdot \mathbb{T})+\rho^\delta\nabla \cdot u(t,x).
\end{split}
\end{equation*}
Integrating by parts, we have
\begin{equation} \label{eq:3.511}
J_1=-2(t+\kappa)\int_{\mathbb{R}^d} x \cdot (\nabla \cdot \mathbb{T}) \text{d}x=2(t+\kappa)\int_{\mathbb{R}^d} \rho^\delta \nabla \cdot u(t,x) \text{d}x.\end{equation}
Now we check $J_2$.
\begin{equation} \label{eq:5.8}
\begin{split}
\partial_t \left( \frac{1}{2}\rho|u|^2+\frac{p_m}{\gamma-1}+\widetilde{E}_r(t) \right)&=-\nabla \cdot \left(\frac{1}{2}\rho|u|^2u\right)-\frac{\gamma}{\gamma-1}\nabla \cdot (p_m u)+u\cdot (\nabla \cdot \mathbb{ T})\\
&-\nabla\cdot F_r+\int_0^\infty \text{d}v \int_{S^{d-1}}\left( -1+\frac{u \cdot \Omega}{c}\right)K_a\cdot(I-\overline{B}(v)) \text{d}\Omega \text{d}v,
\end{split}
\end{equation}
then we have
\begin{equation}\label{k7}
\begin{split}
J_2=2(t+\kappa)^2 \int_{\mathbb{R}^d} \left(u\cdot (\nabla \cdot \mathbb{ T})+\int_0^\infty \text{d}v \int_{S^{d-1}}\Big( -1+\frac{u \cdot \Omega}{c}\Big)K_a\cdot(I-\overline{B}(v))\text{d}\Omega \text{d}v\right)\text{d}x.
\end{split}
\end{equation}
From (\ref{eq:bbb}) and Cauchy's inequality, we know that
\begin{equation}\label{eq:5.9}
\begin{split}
\nabla \cdot (u \mathbb{T })
=&\sum_{i=1}^{d}\sum_{j=1}^{d}\left( u_{i}\frac{\partial{\mathbb{T }_{ij}}}{\partial{x_j}}+\partial_{j}u_{i}\mathbb{T }_{ij}\right)\\
=&u\cdot (\nabla \cdot \mathbb{T }) +\sum_{i=1}^{d}\sum_{i=1}^{d} \partial_{j}u_{i} \Big(\mu(\rho)  (\partial_{j}u_{i}+\partial_{i}u_{j})+\lambda(\rho)\delta_{ij}\nabla \cdot u\Big)\\
=&u\cdot (\nabla \cdot \mathbb{T })+2\mu(\rho)\sum_{i=1}^{d}(\partial_{i}u_{i})^{2}+\mu(\rho)  \sum_{i\neq j}^{d} (\partial_{i}u_{j})^2\\
&+2\mu(\rho) \sum_{i>j} (\partial_{i}u_{j})(\partial_{j}u_{i})+\lambda(\rho) \left(\sum_{i=1}^{d}\partial_{i}u_{i}\right)^{2}\\
\geq& u\cdot (\nabla \cdot \mathbb{T })+\big(\lambda(\rho)+\frac{2}{d}\mu(\rho)\big)(\nabla\cdot u)^2=u\cdot (\nabla \cdot \mathbb{T })+\rho^\delta (\nabla\cdot u)^2.
\end{split}
\end{equation}
We integrate (\ref{eq:5.9}) to get
\begin{equation}\label{eq:5.10}
\int_{\mathbb{R}^d} u\cdot (\nabla \cdot \mathbb{T })\text{d}x\leq -\int_{\mathbb{R}^d}\rho^\delta (\nabla\cdot u)^2\text{d}x.
\end{equation}
From (\ref{gh1}) and Lemma \ref{lemma:3.1}, we know that
if $ T\leq T_c=\frac{2R_0}{c}$,  the proof is finished.\\
If $ T > T_c $, $\forall x \in S_p(t)\subseteq B_{R_0}$, we have
\begin{equation}\label{key2}
Q_r(t)=\int_0^\infty \text{d}v \int_{S^{d-1}}\left( -1+\frac{u \cdot \Omega}{c}\right)K_a\cdot(I-\overline{B}(v))\text{d}\Omega \text{d}v=0.
\end{equation}
Then from (\ref{k7}),  (\ref{eq:5.10}) and (\ref{key2}), we have
\begin{equation}\label{k9}
J_2\leq -2(t+\kappa)^2\int_{\mathbb{R}^d}\rho^\delta (\nabla\cdot u)^2\text{d}x,\ \forall t\in (T_c,T).
\end{equation}
From Lemma \ref{lemma:3.1}, along with (\ref{eq:5.6}), (\ref{eq:3.511}) and (\ref{k9}), for $T_c< t< T$, we get
\begin{equation}\label{eq:5.13}
\begin{split}
\frac{d}{\text{d}t}I_r(t)\leq & \frac{2}{\gamma-1}(2-d(\gamma-1))(t+\kappa)\int_{B_{R_0}}p_m \text{d}x\\
&-2(t+\kappa)^2\int _{B_{R_0}}\rho^\delta(\nabla\cdot u)^2\text{d}x+2(t+\kappa)\int_{B_{R_0}} \rho^\delta \nabla \cdot u(t,x) \text{d}x.
\end{split}
\end{equation}
According to the Cauchy's inequality and Young's inequality, we have
\begin{equation} \label{11111}
\begin{split}
&-2(t+\kappa)^2\int_{B_{R_0}}\rho^\delta(\nabla\cdot u)^2\text{d}x+2(t+\kappa)\int_{B_{R_0}} \rho^\delta \nabla \cdot u(t,x) \text{d}x\\
\leq & -2(t+\kappa)^2\int _{B_{R_0}}\rho^\delta(\nabla\cdot u)^2\text{d}x+(t+\kappa)^2\int _{B_{R_0}}\rho^\delta(\nabla\cdot u)^2\text{d}x+\int_{B_{R_0}}\rho^\delta \text{d}x\\
\leq & \int_{B_{R_0}}\rho^\delta \text{d}x \leq \frac{\delta}{\gamma}\int_{B_{R_0}}\rho^\gamma \text{d}x+\frac{\gamma-\delta}{\gamma}|B_{R_0}|.
\end{split}
\end{equation}
Combining (\ref{eq:5.13}) and (\ref{11111}), for $ T > T_c $ we have
\begin{equation}\label{eq:3.1331}
\begin{split}
\frac{d}{\text{d}t}I_r(t) \leq & \frac{2}{\gamma-1}(2-d(\gamma-1))(t+\kappa)\int_{B_{R_0}}p_m \text{d}x+\frac{\delta}{\gamma}\int_{B_{R_0}}\rho^\gamma \text{d}x+\frac{\gamma-\delta}{\gamma}|B_{R_0}|.
\end{split}
\end{equation}
We  look  back at $ I_r(t) $ of (\ref{def}).
\begin{equation}\label{eq:3.131}
\begin{split}
\frac{2-d(\gamma-1)}{t+\kappa}I_r(t)=&\frac{2}{\gamma-1}\big(2-d(\gamma-1)\big)\int_{B_{R_0}} |x-(t+\kappa)u|^{2}\rho \text{d}x\\
&+\frac{2(2-d(\gamma-1))}{\gamma-1}(t+\kappa)\int_{B_{R_0}} p_m \text{d}x.
\end{split}
\end{equation}
In the case $1< \gamma < 1+\frac{2}{d}$, comparing (\ref{eq:3.1331}) and (\ref{eq:3.131}), we have
\begin{equation}\label{eq:3.141}
\frac{d}{\text{d}t}I_r(t) \leq \frac{2-d(\gamma-1)}{t+\kappa} I_r(t)+\frac{\delta(\gamma-1)}{2\gamma(t+\kappa)^2}I_r(t)+\frac{\gamma-\delta}{\gamma}|B_{R_0}|,
\end{equation}
where $ T_c< t< T $. From Lemma \ref{lemma:3.3}, we have $T< +\infty$.

In the case $ 1+\frac{2}{d}\leq \gamma < +\infty $,
due to $2-d(\gamma-1)\leq 0$ and $\kappa\geq 1$, from (\ref{eq:3.1331}) we have
\begin{equation}\label{eq:3.1335}
\begin{split}
\frac{d}{\text{d}t}I_r(t) \leq &  \frac{\delta}{\gamma}\int_{B_{R_0}}\rho^\gamma \text{d}x+\frac{(\gamma-\delta)}{\gamma}|B_{R_0}|\\
\leq &\frac{\delta(\gamma-1)}{2\gamma(t+\kappa)^2}I_r(t)+\frac{\gamma-\delta}{\gamma}|B_{R_0}|.
\end{split}
\end{equation}
Because $\frac{\delta(\gamma-1)}{2\gamma}<2$, similarly to the proof of Lemma \ref{lemma:3.3}, we still have $T< +\infty$.
\end{proof}

\section{Remarks on one-dimensional case}
In one-dimensional space, the system (see \cite{BD} \cite{BS}) is not obtained directly via letting $d=1$ in system (\ref{eq:2.1}). Consider the three-dimensional case that the specific radiation intensity $I(t,x,v,\Omega)$ ($x=(x_1,x_2,x_3)$) only depends  upon  the single spatial coordinate $x_3$ and the single angular coordinate $\phi$, the angle between $\Omega$ and $x_3$ axis. Introducing $\omega=cos\phi$, since $I=I(t,x_3,v,\omega)$, we have
\begin{equation}
\label{eq:1.qq}
\Omega \cdot \nabla I(t,x_3,v,\omega)=\Omega_3 \partial_{x_3}I(t,x_3,v,\omega)=\omega \partial_{x_3}I(t,x_3,v,\omega).
\end{equation}
So, the one-dimensional radiation hydrodynamics equations read as (\cite{gp})
\begin{equation}
\label{eq:1.qq}
\begin{cases}
\displaystyle
\frac{1}{c}\partial_tI+\omega\partial_x I=-K_a\cdot(I-\overline{B}(v)),\\[10pt]
\displaystyle
\partial_{t}\rho+\partial_x(\rho u)=0,\\[10pt]
\displaystyle
\partial_{t}(\rho u)+\partial_x(\rho u^2+p_m)
   =\mu \partial_{xx}u+\frac{1}{c}\int_0^\infty \int_{-1}^{1} K_a\cdot(I-\overline{B}(v))\omega d\omega \text{d}v,\\[10pt]
 \displaystyle
(\partial_tS+u\cdot \nabla S)p_m=\widetilde{N}_r,
\end{cases}
\end{equation}
where
$$
\widetilde{N}_r=(\gamma-1)\Big(\int_0^\infty \int_{-1}^{1} \left(1-\frac{u\omega}{c}\right)K_a\cdot(I-\overline{B}(v)) d\omega \text{d}v
+(\partial_xu)^2+\partial_x(k(\theta)\partial_x\theta)\Big).
$$
We emphasize that the travel direction $\omega\in [-1,1]$, and
\begin{equation}
\label{eq:1.31}
\begin{cases}
\displaystyle
E_r=\frac{1}{c}\int_0^\infty \int_{-1}^{1} I(t,x,v,\omega)d\omega \text{d}v,\\[10pt]
\displaystyle
F_r=\int_0^\infty \int_{-1}^{1}  I(t,x,v,\omega)\omega d\omega \text{d}v,\\[10pt]
\displaystyle
P_r=\frac{1}{c}\int_0^\infty \int_{-1}^{1}  I(t,x,v,\omega)\omega^2 d\omega \text{d}v.
\end{cases}
\end{equation}
We consider the Cauchy problem of  (\ref{eq:1.qq})  with the initial data
\begin{equation} \label{eq:2.2q}
I|_{t=0}=I_0(x,v,\omega), \quad (\rho,u,S)|_{t=0}=(\rho_0(x),u_0(x), S_0(x))
\end{equation}
satisfying
\begin{equation} \label{eq:5u}
I_0(x,v,\omega)\in L^2( \mathbb{R}^+ \times [-1,1]; H^s(\mathbb{R})),\ (\rho_0,u_0,S_0)(x)\in H^s(\mathbb{R}),
\end{equation}
\begin{equation} \label{eq:2.2w}
I_0 \geq \overline{B}(v) \ \text{for} \ (x,v,\omega)\in \mathbb{R}\times \mathbb{R}^{+}\times [-1,1]; \ I_0 \equiv\overline{B}(v),\ \forall \ |x|\geq R_0.
\end{equation}
\begin{remark}\label{r55}
From system (\ref{eq:1.31}), we know that $|\omega|\leq 1$ in one-dimensional case instead of $|\Omega|=1$ in multi-dimensional case. However, according to the proofs in Sections 2 and 3, there are some difficulties due to $|w|\leq 1$ instead of $|\omega|=1$ for the proofs of Proposition \ref{pro:2.1} and Theorem \ref{th:2.1}, since we cannot get a valid estimate for the increasing of the material entropy.
In order to avoid this difficulty, we assume that the initial specific intensity of radiation $I_0$ has some directional conditions as shown in Theorem \ref{th:2.0000} such that we can prove the desired conclusions via the property of the total material energy obtained in Section $2.2$. So we can also extend some conclusions for multi-dimensional case to one-dimensional model and have the following theorem.
\end{remark}

\begin{theorem}\label{th:2.1w}
The conclusions obtained in Lemmas \ref{lemma:2.1}-\ref{lemma:2.5}, Theorems \ref{th:2.0000}-\ref{th:2.20}, Theorems \ref{th:3.1} and \ref{co:2.1} are all true for the Cauchy problem (\ref{eq:1.qq})-(\ref{eq:2.2q}).
\end{theorem}

\begin{theorem}[\textbf{One-dimensional isentropic flow}]\label{th:2.1ww} The
conclusions obtained in  Theorems \ref{th:2.0000}-\ref{th:2.20}, Theorems \ref{th:3.1} and \ref{co:2.1} are all true for the corresponding  Cauchy problem for one-dimensional isentropic flow.
\end{theorem}
We omit the proofs here, since they are similar to the multi-dimensional case.

\end{document}